\documentclass[11pt]{amsart}

\usepackage[utf8]{inputenc}
\usepackage[T1]{fontenc}
\usepackage{amssymb,amsmath}
\usepackage{graphics}
\usepackage{xcolor}
\usepackage{enumitem}
\usepackage[colorlinks = true, linkcolor = teal, urlcolor  = teal, citecolor = violet]{hyperref}
\usepackage[margin=1.0in]{geometry}
\usepackage{cleveref}
\crefname{equation}{Equation}{Equations}
\usepackage{bbm}

\usepackage{listings}

\definecolor{codegreen}{rgb}{0,0.6,0}
\definecolor{codegray}{rgb}{0.5,0.5,0.5}
\definecolor{codepurple}{rgb}{0.58,0,0.82}
\definecolor{backcolour}{rgb}{0.95,0.95,0.92}

\lstdefinestyle{mystyle}{
    backgroundcolor=\color{backcolour},   
    commentstyle=\color{codegreen},
    keywordstyle=\color{magenta},
    numberstyle=\tiny\color{codegray},
    stringstyle=\color{codepurple},
    basicstyle=\ttfamily\footnotesize,
    breakatwhitespace=false,         
    breaklines=true,                 
    captionpos=b,                    
    keepspaces=true,                 
    numbers=left,                    
    numbersep=5pt,                  
    showspaces=false,                
    showstringspaces=false,
    showtabs=false,                  
    tabsize=2
}

\makeatletter
\newcommand{\pushright}[1]{\ifmeasuring@#1\else\omit\hfill$\displaystyle{#1}$\fi\ignorespaces}
\newcommand{\pushleft}[1]{\ifmeasuring@#1\else\omit$\displaystyle{#1}$\hfill\fi\ignorespaces}
\makeatother

\newtheorem{theorem}{Theorem}[section]
\newtheorem*{theorem*}{Theorem}
\newtheorem{definition}[theorem]{Definition}
\newtheorem*{definition*}{Definition} 
\newtheorem{proposition}[theorem]{Proposition}
\newtheorem{corollary}[theorem]{Corollary}

\newtheorem{lemma}[theorem]{Lemma}
\newtheorem{remark}[theorem]{Remark}

\AddToHook{env/lemma/begin}{\crefalias{theorem}{lemma}}
\AddToHook{env/proposition/begin}{\crefalias{theorem}{proposition}}
\AddToHook{env/definition/begin}{\crefalias{theorem}{definition}}
\AddToHook{env/corollary/begin}{\crefalias{theorem}{corollary}}
\AddToHook{env/remark/begin}{\crefalias{theorem}{remark}}

\newcommand{\id}{{\mathrm{Id}}}

\newcommand{\range}{\operatorname{range}}
\newcommand{\tr}{{\operatorname{tr}}}

\renewcommand{\d}{{\,\rm d}}

\newcommand{\one}{{\mathbf 1}}

\newcommand{\ee}{{\mathbb E}}

\newcommand{\nn}{{\mathbb N}}

\newcommand{\cc}{{\mathbb C}}
\renewcommand{\P}{{\mathrm P}}
\newcommand{\supp}{\operatorname{supp}}
\newcommand\pc[1]{{\P(\cc^{#1}) }}

\newcommand\pcd{{\P(\cc^{d}) }}

\newcommand{\Ld}{{M_d(\cc)}}

\newcommand{\inv}{\textnormal{inv}}
\newcommand{\unif}{\textnormal{unif}}

\newcommand{\GAP}{\operatorname{GAP}}

\makeatletter
\newcommand{\specialcell}[1]{\ifmeasuring@#1\else\omit$#1$\ignorespaces\fi}
\makeatother

\usepackage[backend=bibtex,giveninits=true,citestyle=alphabetic,bibstyle=alphabetic,maxbibnames =100,maxcitenames=5,isbn=false]{biblatex}
\renewbibmacro{in:}{}
\begin{document}
    
    \title[Uniform quantum trajectories]{Invariant measures \\ of randomized quantum trajectories}
    
    \author{Tristan Benoist}
    \address{Institut de Math\'ematiques de Toulouse, Université de Toulouse, CNRS, 
    UPS, F-31062 Toulouse Cedex 9, France}
    \email{tristan.benoist@math.univ-toulouse.fr}
    \author{Sascha Lill}
    \address{Department of Mathematical Sciences, Universitetsparken 5, DK-2100 Copenhagen, Denmark}
    \email{sali@math.ku.dk}   
    \author{Cornelia Vogel}
    \address{Department of Mathematics, LMU Munich, Theresienstr. 39, 80333 Munich, Germany}
    \email{cornelia.vogel@math.lmu.de}
    
\setcounter{tocdepth}{1}
    
    
    \begin{abstract} 
    Quantum trajectories are Markov chains modeling quantum systems subjected to repeated indirect measurements. Their stationary regime depends on what observables are measured on the probes used to indirectly measure the system. In this article we explore the properties of quantum trajectories when the choice of probe observable is randomized. The randomization induces some regularization of the quantum trajectories. We show that non-singular randomization ensures that quantum trajectories purify and therefore accept a unique invariant probability measure. We furthermore study the regularity of that invariant measure. In that endeavour, we introduce a new notion of ergodicity for quantum channels, which we call multiplicative primitivity. It is a priory stronger than primitivity but weaker than positivity improving. Finally, we compute some invariant measures for canonical quantum channels and explore the limits of our assumptions with several examples.
    \end{abstract}
    
    \date{\today}
    
    \maketitle
    \thispagestyle{empty}   
 
{\small \tableofcontents} %

    \section{Introduction} 
    Quantum systems subjected to repeated indirect measurements evolve randomly according to stochastic processes called quantum trajectories. They typically model experiments in quantum optics -- see for example \cite{wisemanmilburn,haroche2006exploring}. In \cite{MaaKumm}, Maassen and Kümmerer proved that the state of the system has a tendency to purify along the quantum trajectory, meaning that, if it starts in a mixed state, it has a tendency to get closer to the set of pure states, reaching it asymptotically under the right conditions. That property was leveraged in \cite{BenFra} to prove uniqueness of the invariant measure for quantum trajectories. The proof of that fundamental result required a new approach since standard techniques used for Markov chains are not efficient for quantum trajectories. Indeed, they are neither $\varphi$-irreducible nor contracting in general -- see \cite[Section~8]{BenFat}. While inspired by works on random products of matrices -- especially \cite{guivarc2016spectral} -- new techniques were developed to deal with non-invertible and non strongly irreducible matrices. Following \cite{BenFra}, several finer limit theorems have been proved in \cite{BenFat,BHP24}. The proofs of \cite{BHP24} were obtained through the establishment of a spectral gap for the Markov kernel. In \cite{benoist2024dark}, the set of invariant measures has been classified when purification does not hold.

    \medskip
    The average evolution of a quantum trajectory results from the repeated application of a quantum channel $\Phi^*$\footnote{In this article, $\Phi^*$ denotes the completely positive, trace-preserving map acting on density matrices $\rho$, while its dual $\Phi$ is a unit-preserving map acting on observables. We call both maps ``quantum channels''.} to the initial density matrix state of the system. A first condition required for the uniqueness of the invariant measure is that $\Phi^*$ accepts a unique invariant state. As already mentioned, the second assumption used to prove uniqueness in~\cite{BenFra} is that the quantum trajectory purifies. That property not only depends on $\Phi^*$, but also on a Kraus decomposition that is chosen for $\Phi^*$: by the Kraus decomposition theorem, we may always write
    $$\Phi^*(\rho) = \sum_{i=1}^k v_i\rho v_i^*,$$
    for $k \in \nn$ large enough and $(v_i)_{i=1}^k$ some indexed family of matrices. This choice of matrices is not unique. It corresponds to a choice of observable, or more precisely orthonormal basis, measured on the probe used to indirectly measure the system -- see \cite[Theorem~7.14]{busch2016quantum}. Depending on the choice of $(v_i)_i$, for a fixed $\Phi^*$, purification may or may not happen.
    
    \medskip
    In the present article, we explore the properties of quantum trajectories when the choice of decomposition is randomized at each time step. In particular, we are interested in the agnostic position characterized by choosing uniformly the orthonormal basis along which the probe is measured. This represents a canonical example of a quantum trajectory that gains regularity from the randomization. All our results assume the quantum channel $\Phi^*$ is irreducible, so it accepts a unique invariant state of full rank.
    
    Our first main result concerns conditions for purification. In \Cref{thm:purification}, we prove that if the randomization is not singular, then purification holds. Thanks to \cite[Theorem~1.1]{BenFra}, this implies uniqueness of the invariant measure. Actually, in \Cref{lem:IC implies pur}, we prove more generally that any informationally complete instrument leads to purification.
    
    Our second main result is that, contrary to the general case as proved in \cite[Section~8]{BenFat}, sufficiently randomized quantum trajectories are $\varphi$-irreducible with respect to the uniform measure on the set of states -- see \Cref{thm:phi-irred}. If moreover a matrix $v_i$ is invertible, then the invariant measure is equivalent to the uniform one. To prove these results we introduce a new notion of ergodicity for quantum channels we call multiplicative primitivity -- see \Cref{def:mPrim}. It is a priory stronger than primitivity (\emph{i.e.}, irreducible and aperiodic) but weaker than positivity improving ($\Phi^*(\rho)>0$ for any $\rho\geq 0$) -- see \Cref{prop:mPrim to Prim,prop:pos improve implies mPrim}. In dimension $2$, multiplicative primitivity and primitivity are equivalent -- see \Cref{prop:2D}.
    
    Our third main result concerns the special case of a uniformly randomized orthonormal basis measured on the probe. Here, symmetries of the channel $\Phi^*$ translate into symmetries of the invariant measure of the quantum trajectory -- see \Cref{thm:symmetries}. The proof relies on the expression of the Markov chain kernel as a GAP measure as introduced in \cite{JRW94} (see also \cite{GLTZ06,GLMTZ16}).
    
    We illustrate our results with some examples. First, we provide trivial examples of invariant measures: One is the GAP measure for some state $\rho$ and the other is the uniform measure. Second, in dimension $2$, we explicitly state the integral equation satisfied by the invariant probability density when the choice of measured probe's orthonormal basis is uniform. The equation is a consequence of the explicit density expression for GAP measures from \cite[Equation~(18)]{GLTZ06}, whose validity requires a nontrivial proof that in dimension $2$ the matrices $v_i$ are almost always invertible. We then provide two examples of multiplicatively primitive channels in dimension $3$, one with invertible matrices and the other with only non-invertible matrices.
    
    These results raise several questions. First: is the new notion of multiplicative primitivity necessary to prove $\varphi$-irreducibility? Second: can one find primitive counter-examples to multiplicative primitivity in dimension $3$ or higher? Finally, can an explicit expression of the invariant measure density be derived, at least for some simple nontrivial channels? More generally the present work motivates further study of quantum trajectories defined by randomized measurements.

    \medskip
    The article is structured as follows. In \Cref{sec:setup}, we provide the mathematical setup and state our main assumptions. In \Cref{sec:main}, we state our main results. In \Cref{sec:mPrim}, we discuss the new multiplicative primitivity assumption. 
    In \Cref{sec:proofs} we gather all the proofs of our main results. In \Cref{sec:xpl}, we provide different examples of channels for which we can compute the invariant measure. We also discuss the specific case of dimension $2$ and we provide some mathematically relevant examples in dimension $3$. In Appendix~\ref{app:irr}, we recall some equivalent characterizations of different notions of ergodicity for quantum channels. In Appendix~\ref{app:instru}, we derive the space of operators generated by informationally complete instruments. In Appendix~\ref{app:GAP}, we summarize some properties of GAP measures useful to us. Finally, in Appendix~\ref{app:code}, we provide the SageMath~10 code we used to prove that our examples in dimension $3$ are multiplicatively primitive.

    \section{Setup and assumptions}\label{sec:setup}
    Consider the set $\Ld$ of $d\times d$ complex matrices equipped with its Borel $\sigma$-algebra. Let $\Phi:\Ld\to\Ld$ be a completely positive unit preserving map, \emph{i.e.}, a quantum channel. By Stinespring's dilation theorem, or the Kraus decomposition theorem, there exists a finite tuple $(v_i)_{i=1}^k \in \Ld^k$ such that
    $$\Phi(X) = \sum_{i=1}^k v_i^*Xv_i.$$
    The choice of $k$ and $(v_i)_{i=1}^k$ is not unique. Fixing $k$, following Stinespring's theorem, any $k$-tuple of matrices $(w_i)_{i=1}^k$ also satisfies
    $$\Phi(X) = \sum_{i=1}^k w_i^*Xw_i ,$$
    if and only if there exist a $u \in U(k)$, with $U(k)$ denoting the Lie group of $k \times k$ unitary matrices, such that
    $$w_i=\sum_{j=1}^k u_{ij}v_j.$$
    A set of such matrices $\{w_i\}_i$ is called a Kraus decomposition of $\Phi$.
    
    The freedom in the choice of Kraus decomposition implies that, for any probability measure $\lambda$ on $U(k)$ equipped with its Borel $\sigma$-algebra, denoting $v_i(u)=\sum_{j=1}^k u_{ij}v_j \in \Ld$,
    $$\Phi(X) = \sum_{i=1}^k\int_{U(k)}v_i(u)^* X v_i(u)\d\lambda(u).$$
    The sum over $i$ can be absorbed into the probability measure: for $i\in\{1,\dotsc,k\}$, let $u(i) \in U(k)$ with first row being $u(i)_{1l}=\delta_{l,i}$, and let $\lambda_i$ be the image measure of $\lambda$ by the map $u\mapsto u(i)u$. Then, with 
    \begin{align}
        \label{eq:def mu}
        \mu := \sum_{i=1}^k\lambda_i,
    \end{align}
    we may write
    $$\Phi(X)=\int_{U(k)} v(u)^*Xv(u) \d\mu(u),$$
    where $v: U(k) \to \Ld$ is defined by $v(u)= v_1(u)$. Note that $\mu$ is such that $\mu(U(k)) = k$. It ensures that, for any fixed $x \in \cc^d$ with $\|x\|=1$, $\|v(u)x\|_2^2\d\mu(u)$ defines a probability measure on $U(k)$.

    Seeing $\Phi$ as the result of a unitary interaction between a system and a probe -- see \cite[Theorem~7.14]{busch2016quantum} -- the measure $\lambda$ can be interpreted as a random choice of observable measured on the probe. Then, conditioning on the measurement result and the detector adjustment, both of them summarized in $ u \in U(k) $, the evolution of the state vector also becomes random.

    \medskip
    We describe state vectors $\hat x$ as elements of the projective space $\pcd$, which consists of equivalence classes $\hat x = \{z x: z\in \cc\} \in \pcd$ for any $x\in \cc^d\setminus\{0\}$. Given $\hat x\in \pcd$, we will denote by $x$ an arbitrary norm $1$ representative of $\hat x$. We equip $\pcd$ with its Borel $\sigma$-algebra and $\nu_\unif$ denotes the uniform probability measure on $\pcd$. Further, for $w\in \Ld$ and $\hat x\in \pcd$ such that $wx\neq 0$, we denote $w\cdot \hat x=\widehat{wx}$. If $wx=0$ we set $w\cdot\hat x$ arbitrarily. Then $\nu_\unif$ is the unique probability measure on $\pcd$ such that $u\cdot \hat x$ has the same law as $\hat x$ for any unitary matrix $u$.
    
    We also often rely on Dirac's notation, where a vector $x$ in a Hilbert space is denoted $|x\rangle$ and its dual by $\langle x|$. So, $|x\rangle$ is not necessarily normalized and $\langle x|\,| x'\rangle=\langle x,x'\rangle$ is the scalar product of the Hilbert space. It follows that, $\langle w x||wx\rangle=\langle x| w^*w|x\rangle=\|wx\|^2$. In this notation, $\pcd$ is in bijection with the set of rank-one orthogonal projectors on $\cc^d$ (pure state density matrices) through the map $\hat x\mapsto |x\rangle\langle x|$.
    
    \medskip
    Given an initial datum $\hat x_0 \in \pcd$ and the above-described measure $\mu$ on $U(k)$, the random evolution of the state vector is then described by a Markov chain $(\hat x_n)_{n \in \nn_0}$\footnote{We use the convention $\nn=\{1,2\dotsc\}$ and $\nn_0=\{0\}\cup \nn$.} with
    $$\hat x_{n+1}=v(U_n)\cdot \hat x_n\quad \mbox{with }U_n\sim \|v(u)x_n\|_2^2\d\mu(u).$$
    This process is what is called a quantum trajectory.

    The Markov kernel of this chain is given, for any bounded measurable function $f:\pcd\to\cc$, by
    $$\Pi f(\hat x)=\ee(f(\hat x_1)|\hat x_0=\hat x)=\int_{U(k)}f(v(u)\cdot\hat x)\|v(u)x\|_2^2\d\mu(u).$$
    For probability measures $\nu$ over $\pcd$, $\nu \Pi$ denotes the probability measure defined by $\ee_{\nu\Pi}(f)=\ee_\nu(\Pi f)$ for any bounded measurable function $f$.
    In terms of ensembles, $\Pi(\hat x,A)=\Pi\one_A(\hat x)$ for any measurable subset $A$ of $\pcd$, where $\one_A$ is the characteristic function of $A$.

    If $\lambda=\delta_{u_0}$ for some $u_0\in U(k)$, then the quantum trajectory corresponds to the evolution of the quantum system when a fixed observable is measured on the probes. As already mentioned, by contrast, taking $\lambda$ to not be an extreme point in the set of probability measures corresponds to a randomization of the observable measured on the probes. In this article we focus on such randomizations, and especially on those ones that have some regularity with respect to a Haar measure on $U(k)$. 
    
    \begin{definition}[Non-singular $\mu$]\label{def:non singular}
    Let $\mu_{\unif}$ be a Haar measure over $U(k)$. We say that $\mu$ is non-singular if $\mu_{\unif} (\supp\mu) > 0$.
    \end{definition}
    
    The choice $\mu=\mu_\unif$ corresponds to a random uniform choice of the basis measured on the probe at each time step. It corresponds to an agnostic position with respect to the probe measurement observable.
    
    \medskip
    The only other assumptions we use concern the ergodic properties of the quantum channel $\Phi$.
   \begin{definition}[Irreducible channels]\label{def:irr}
    A quantum channel $\Phi$ is called irreducible if there exists $n\in \nn$ such that for any positive semi-definite $X\in \Ld$, $(\id_\Ld+\Phi)^n(X)$ is positive definite.
   \end{definition}
    This is the standard assumption leading to the Perron--Frobenius theorem for positive maps on finite-dimensional $C^*$-algebras -- see~\cite{Evans1978}. A stronger assumption requires that $\Phi^n$ maps positive semi-definite matrices to positive definite ones. It is called primitivity and implies irreducibility.
    \begin{definition}[Primitive channels]\label{def:prim}
    A quantum channel $\Phi$ is called primitive if there exists $n\in \nn$ such that for any positive semi-definite $X\in \Ld$, $\Phi^n(X)$ is positive definite.
    \end{definition}
    There exists a notion of period for irreducible quantum channels -- see \cite{Evans1978}. Primitive maps are actually the aperiodic irreducible maps -- see the same reference.
    
    For some of our results, we assume a new, stronger notion of primitivity, which we call multiplicative primitivity. Given a quantum channel $\Phi$, for $p \in \mathbb{N}$ we define
    \begin{equation} \label{eq:cVp}
        \mathcal V_p := \operatorname{linspan}\{v_{i_1} \dotsb v_{i_p} : (i_1, \ldots, i_p) \in \{1,\ldots,k\}^p\} \subset M_d(\cc)
    \end{equation}
    where the matrices $(v_i)_{i=1}^k$ are a Kraus decomposition of $\Phi$. Note that, by linearity, $\mathcal V_p$ does not depend on the particular choice of Kraus decomposition.
    From \Cref{thm:prim generated subspace}, $\Phi$ is primitive if and only if for any non-zero $x \in \cc^d$, we get $\mathcal V_p x = \cc^d$ for some $p \in \mathbb{N}$. Our stronger form of primitivity is the following.
    \begin{definition}[Multiplicative primitivity]\label{def:mPrim}
    A quantum channel $\Phi$ is called multiplicatively primitive if for any $\hat x\in \pcd$, there exists $p\in \nn$ such that
    $$\mathcal V_1^px :=\{ a_p\dotsb a_1x : a_1,\dotsc,a_p\in \mathcal V_1\} \qquad \textnormal{satisfies} \qquad \mathcal V_1^px=\cc^d.$$
    \end{definition}
    Since $\mathcal V_1$ does not depend on the choice of Kraus decomposition, multiplicative primitivity is indeed a property of $\Phi$. Since $\mathcal V_1^p \subseteq \mathcal V_p$, \Cref{thm:prim generated subspace} yields that multiplicative primitivity implies primitivity. The validity of the reverse implication in $d \ge 3$ is, as far as we know, an open question, while for $d = 2$, we prove this reverse implication in \Cref{prop:mPrim dim 2}. We further discuss multiplicative primitivity in \Cref{sec:mPrim}.

    We are now equipped to state our main results.

\section{Main results}\label{sec:main}

    \subsection{Uniqueness of the invariant measure}
    Our first result is the uniqueness of the invariant measure. To formulate it we introduce the $1$-Wasserstein distance
    $$W_1(\nu_1,\nu_2)=\inf_{P}\int d(\hat x,\hat y)\d P(\hat x,\hat y) ,$$
    where $d(\hat x,\hat y)=\sqrt{1-|\langle x,y\rangle|^2}$ is a metric over $\pcd$ and the infimum is taken over all probability measures $P$ on $\pcd\times\pcd$ such that $P(\ \cdot\ ,\pcd)=\nu_1$ and $P(\pcd,\ \cdot\ )=\nu_2$.
    \begin{theorem}
        \label{thm:uniqueness invariant}
        Assume $\Phi$ is irreducible and $\mu$ is non-singular, see ~\Cref{def:irr,def:non singular}. Then, $\Pi$ accepts a unique invariant measure $\nu_{\inv}$ on $\pcd$, and there exist $m\in \{1,\dotsc,d\}$ and two constants $C>0$ and $\gamma\in[0,1)$ such that for any probability measure $\nu$ on $\pcd$ and any $n\in \nn$,
        $$W_1\Bigg( \tfrac{1}{m}\sum_{l=0}^{m-1}\nu\Pi^{mn+l},\nu_{\inv} \Bigg) \leq C \gamma^n.$$
    \end{theorem}
    This theorem is a corollary of \Cref{thm:purification} and \cite[Theorem~1.1]{BenFra}. Indeed, as discussed on page~310 of \cite{BenFra}, the irreducibility assumption of \Cref{def:irr} is equivalent to assumption {\bf ($\phi$-erg)} in \cite{BenFra} with $E=\cc^d$. We provide a detailed proof of this equivalence in \Cref{thm:irr invariant subsapce}.
    
    The only hypothesis missing to apply \cite[Theorem~1.1]{BenFra} is purification. It is a central notion in the study of quantum trajectories. We formulate the usual equivalent assumption for a general measure space, as \Cref{lem:IC implies pur} in the proof uses it to provide an original result of independent interest.
    \begin{definition}[Purification]\label{def:purification}
        Let $(\Omega,\kappa)$ be a measure space. Let $v:\Omega\to M_d(\cc)$ be a measurable map such that $\int_\Omega v^*(\omega)v(\omega) \d \kappa(\omega)=\id_{\cc^d}$. Then, the measure $\kappa$ is said to purify if all the orthogonal projectors $\pi$ satisfying that for any $p\in \nn$ and $\kappa^{\otimes p}$-almost any $(\omega_1,\dotsc,\omega_p)$,
        \begin{equation}
            \label{eq:purification}
            \pi v(\omega_1)^*\dotsb v(\omega_p)^*v(\omega_p)\dotsb v(\omega_1)\pi=\|v(\omega_p)\dotsb v(\omega_1)\pi\|^2\pi,
        \end{equation}
        are of rank at most $1$.
    \end{definition}
    In \cite{BenFra}, purification is formulated slightly differently: \Cref{eq:purification} is required to hold for all matrices $v_1,\dotsc,v_p$ in the support of the image measure of $\kappa$ by $v$. However, by continuity, \Cref{eq:purification} holds for this larger set of matrices if and only if it holds for $\kappa^{\otimes p}$-almost every $(\omega_1,\dotsc,\omega_p)$. Hence, the present definition of purification is equivalent to the one of \cite{BenFra}. The next theorem shows that for non-singular $\mu$, purification holds.
    
\begin{theorem}
    \label{thm:purification}
    Assume $\Phi$ is irreducible and $\mu$ is non-singular, see \Cref{def:irr,def:non singular}. Then $\mu$ satisfies the purification assumption of \Cref{def:purification}.
\end{theorem}
This theorem is proved in \Cref{sec:proof pur}.

    \subsection{$\nu_{\unif}$-irreducibility}
The notion of $\varphi$-irreducibility is standard for Markov chains. It is a generalization of the irreducibility criterion for stochastic matrices to an arbitrary space and can be understood as a notion of ergodicity. We refer the reader to \cite{MT} for a comprehensive presentation of the subject.

\begin{definition}[Proposition~4.2.1 in \cite{MT}]\label{def:phi irr}
Let $\mathsf{X}$ be a measurable space. A Markov chain with kernel $\mathsf{P}$ over $\mathsf{X}$ is $\varphi$-irreducible for a measure $\varphi$ if for any $\mathsf{x}\in \mathsf{X}$ and any measurable $\mathsf{A} \subset \mathsf{X}$ with $\varphi(\mathsf{A})>0$, there exists $n\in\nn$ such that
$$\mathsf{P}^n(\mathsf{x},\mathsf{A})>0.$$
\end{definition}

In general, quantum trajectories are not $\varphi$-irreducible (see \cite[Section~8]{BenFat}). However, the following result gives a sufficient criterion for $\varphi$-irreducibility of a  quantum trajectory. We recall that $\mu\gg \mu_\unif$ means that $\mu_\unif$ is absolutely continuous with respect to $\mu$.

\begin{theorem}
    \label{thm:phi-irred}
    Assume $\Phi$ is multiplicatively primitive, see \Cref{def:mPrim}, and $\mu\gg \mu_\unif$. Then, the Markov chain $(\hat x_n)_n$ on $\pcd$ is $\varphi$-irreducible with
    $$\varphi=\nu_\unif.$$
    Moreover, $\nu_{\inv}\gg\nu_{\rm unif}$.
\end{theorem}
The condition $\mu\gg \mu_\unif$ ensures that any ray in  $\mathcal V_1$ is potentially selected through the measurement.

The $\varphi$-irreducibility of the quantum trajectory has numerous consequences -- see \cite{MT}. In particular, since $\pcd$ is compact, the chain $(\hat x_n)$ is positive Harris. It follows that the Wasserstein distance in \Cref{thm:uniqueness invariant} can be switched to the total variation one -- see \cite[Theorem~13.3.4]{MT}. Moreover, the law of large numbers holds for any $L^1(\nu_\inv)$ function -- see \cite[Theorem~17.0.1]{MT}. This contrasts with \cite[Proposition~8.2]{BenFat} which shows that the law of large numbers may fail for bounded discontinuous functions for general quantum trajectories. We will not list all the properties implied by $\varphi$-irreducibility here, as we focus on the symmetries and regularity of the invariant measure $\nu_\inv$. We refer the reader interested into other consequences to \cite{MT}.

An important result that is partly a consequence of the previous theorem is that for $\mu\gg\mu_\unif$ and $v(u)$ invertible, the invariant measure is equivalent to the uniform one. We recall that $\nu_1\sim\nu_2$ means that $\nu_1\ll \nu_2$ and $\nu_1\gg\nu_2$.
\begin{theorem}
    \label{thm:absolute continuity}
    Assume $\Phi$ is multiplicatively primitive, see \Cref{def:mPrim}, $\mu\gg\mu_\unif$ and $v(u)$ is invertible for $\mu$-almost every $u$. Then, $\nu_\inv\sim\nu_\unif$.
\end{theorem}
\begin{remark}
    The assumption of invertibility is crucial to the proof. In \Cref{sec:xpl non equivalent uniform} we provide a counterexample to \Cref{thm:absolute continuity} when $\mu(\{u: \det v(u)=0\})>0$.
\end{remark}

\begin{remark}\label{rem:invertible everywhere}
    When $\mu\sim\mu_\unif$, the assumption of invertibility of $v(u)$ can be simplified: In \Cref{lem:invertible evrywhere} we show that when $\mu\ll\mu_\unif$, it is sufficient to assume invertibility of $v(u)$ for one single $u\in U(k)$.
\end{remark}

\begin{remark}
    In \Cref{sec:xpl 3D non inv} we provide an example of $\Phi$ multiplicatively primitive such that $\det v(u)=0$ for any $u\in U(k)$. More generally, there exist numerous irreducible linear spaces of non-invertible matrices. For example, the set of anti-symmetric matrices in odd dimension, or the adjoint representation of a semisimple complex Lie algebra. See, for example, \cite{lovasz1989singular,draisma2006small} for these and other examples.
\end{remark}

The theorems of this section are proved in \Cref{sec:proof irred}.

\subsection{Uniformly randomized measurement and symmetries}
    We finally focus on the case $\mu=\mu_\unif$. Leveraging Haar measure invariance, we show that the symmetries of $\Phi$ are recovered in the invariant measure.
    \begin{definition}
    A quantum channel $\Phi$ is $U$-covariant for some $U\in U(d)$ if for any $X\in \Ld$,
    $$\Phi(U^*XU)=U^*\Phi(X)U.$$
    The group generated by all the $U\in U(d)$ such that $\Phi$ is $U$-covariant is called the symmetry group of $\Phi$ and is denoted by $G_\Phi$.
    \end{definition}
    \begin{theorem}
        \label{thm:symmetries}
        Assume $\Phi$ is irreducible, see \Cref{def:irr}, and $\mu=\mu_\unif$. Then, for any $U\in G_\Phi$, $\nu_{\rm inv}$ is invariant under the image $\hat x\mapsto U\cdot \hat x$.
    \end{theorem}
    
    From this theorem, we can deduce the invariant measure when the symmetry group is the full unitary group.
    \begin{corollary}
        \label{cor:full symmetry}
        Assume $\Phi$ is irreducible, see \Cref{def:irr}, and $\mu=\mu_\unif$. Assume moreover that $G_\Phi=U(d)$. Then $\nu_{\rm inv}=\nu_\unif$.
    \end{corollary}
    \Cref{thm:symmetries} is proved in \Cref{sec:proof sym}.

\section{On multiplicative primitivity}\label{sec:mPrim}
As it is an uncommon notion of ergodicity for quantum channels, we discuss multiplicative primitivity. For future reference, we explicitly state in \Cref{prop:mPrim to Prim} that multiplicatively primitive quantum channels are also primitive. Moreover, we show that a positivity improving quantum channel is multiplicatively primitive -- see \Cref{prop:pos improve implies mPrim}. So
$$\text{positivity improving} \quad \implies \quad \text{multiplicatively primitive} \quad \implies \quad \text{primitive}. $$
The question whether multiplicative primitivity is equivalent to primitivity is open in dimension $d\geq 3$. In \Cref{prop:mPrim dim 2}, we show equivalence in $d=2$. We have not found any counterexample yet when $d\geq3$.

Before we prove these results, we highlight the difference between primitivity and its multiplicative version using the point of view of purely generated finitely correlated states as introduced in \cite{fannes1992finitely}, also known as matrix product states (MPS). For any $n\in \nn$, let $V_n:\cc^d\to\cc^d\otimes (\cc^k)^{\otimes n}$ be defined by $V_n=\sum_{i_1,\dotsc,i_n=1}^k v_{i_n}\dotsb v_{i_1}\otimes |e_{i_1}\rangle\otimes\dotsb\otimes |e_{i_n}\rangle$ where $(v_i)$ is a Kraus decomposition of $\Phi$ and $(e_i)$ is an orthonormal basis of $\cc^k$. Then, for any $\hat x\in \pcd$,
$|V_n x\rangle$ is an element of $S(\cc^d\otimes(\cc^k)^{\otimes n})$, the unit sphere\footnote{For any Hilbert space $\mathcal H$ we denote by $S(\mathcal H)$ its unit sphere.} of $\cc^d\otimes(\cc^k)^{\otimes n}$. Therefore, primitivity is equivalent to the existence, for any $\hat x\in \pcd$, of an $n \in \nn$ such that
$$\left\{\frac{(\id_{\cc^d}\otimes \langle \Psi_n|)|V_n x\rangle}{\|(\id_{\cc^d}\otimes \langle \Psi_n|)|V_n x\rangle\|} : \Psi_n\in S((\cc^k)^{\otimes n}), \|(\id_{\cc^d}\otimes \langle \Psi_n|)|V_n x\rangle\|>0\right\}=S(\cc^d),$$
while multiplicative primitivity is equivalent to the existence, for any $\hat x\in \pcd$, of an $n \in \nn$ such that
\begin{align*}
    \left\{\frac{(\id_{\cc^d}\otimes \langle \psi_1|\otimes \dotsb \otimes \langle \psi_n|)|V_n x\rangle}{\|(\id_{\cc^d}\otimes \langle \psi_1|\otimes \dotsb \otimes \langle \psi_n|)|V_n x\rangle\|} : (\psi_i)_{i=1}^n\in S(\cc^k)^n, \|(\id_{\cc^d}\otimes \langle \psi_1|\otimes \dotsb \otimes \langle \psi_n|)|V_n x\rangle\|>0\right\}\quad\\
    \pushright{=S(\cc^d)}.
\end{align*}
Hence, primitivity requires that any state in $S(\cc^d)$ can be obtained by projecting $|V_n x\rangle$ along any, possibly entangled, state $|\Psi_n\rangle$ in $S((\cc^k)^{\otimes n})$, whereas its multiplicative version requires that the same state $|\Psi_n\rangle$ is a product state. That suggests the implication stated in the next proposition.
\begin{proposition}\label{prop:mPrim to Prim}
If $\Phi$ is multiplicatively primitive as in~\Cref{def:mPrim}, then it is primitive, see \Cref{def:prim}.
\end{proposition}
\begin{proof}
This is a direct consequence of $\operatorname{linspan}\mathcal V_1^p=\mathcal V_p$ (compare \Cref{eq:cVp}), using \Cref{thm:prim generated subspace}.
\end{proof}
\begin{proposition}
\label{prop:pos improve implies mPrim}
Assume $\Phi$ is positivity improving, that is, $\Phi(X)>0$ for any non-zero $X\geq 0$. Then, $\Phi$ 
is multiplicatively primitive as in~\Cref{def:mPrim} with $p=1$ for any $\hat x\in \pcd$.
\end{proposition}
\begin{proof}
Let $\hat x\in \pcd$. Assume $\mathcal V_1x\neq \cc^d$. Then, since $\mathcal V_1x$ is a subspace of $\cc^d$, there exists some $\hat y\in \pcd$ such that $\mathcal V_1x\perp \cc y$. That implies $\langle y,\Phi(|x\rangle\langle x|)y\rangle=0$ with $|x\rangle\langle x|$ the orthogonal projector onto $\cc x$. So $\mathcal V_1x\neq \cc^d$ contradicts the fact that $\Phi$ is positivity improving. Therefore, $\mathcal V_1x= \cc^d$.
\end{proof}

\begin{proposition}\label{prop:mPrim dim 2}
For $d=2$, a quantum channel on $M_2(\cc)$ is multiplicatively primitive if and only if it is primitive.
\end{proposition}
\begin{proof}
\Cref{prop:mPrim to Prim} proves the forward implication. It remains to prove that primitivity implies multiplicative primitivity. Assume there exists $\hat x\in \mathrm{P}(\cc^2)$ such that for any $p\in \nn$, $\mathcal V_1^px\neq \cc^2$. Then, there exists a sequence $(\hat y_p)_{p\in \nn} \subset \mathrm{P}(\cc^2)$ such that $\mathcal V_1^px=\cc y_p$. That implies in turn that for any $p\in \nn$, $\Phi^p(|x\rangle\langle x|)=|y_p\rangle\langle y_p|$. This contradicts $\Phi^n(|x\rangle\langle x|)>0$ for some $n\in\mathbb N$. Therefore, $\Phi$ is multiplicatively primitive.
\end{proof}
This proof relies on the fact that in dimension $2$, either $\mathcal V_1x=\cc^2$, or it is one-dimensional. In both cases, $\mathcal V_1^2x$ is a vector space. Already in dimension $3$, we lose this property as $\mathcal V_1x$ can be of dimension $2$ and $\mathcal V_1^2x$ may not be a vector space, but a continuous union of vector spaces.

We could not find any examples of primitive quantum channels on $M_3(\cc)$ that are not multiplicatively primitive. Thus, the equivalence between primitivity and multiplicative primitivity in arbitrary dimension remains an open question of interest.

\section{Proof of the main results}\label{sec:proofs}

\subsection{Purification: Proof of \Cref{thm:purification}}
\label{sec:proof pur}
The proof relies on the notion of informationally complete instruments. In the following $M_{k,\geq}(\cc)$ denotes the set of positive semi-definite $k\times k$ complex matrices.
	\begin{definition}[Informationally complete instruments]\label{def:IC}
		Let $(\Omega,\mathcal{F})$ be a measurable space and let $\mathcal{J}:\mathcal{F}\to\mathcal{L}(M_d(\mathbb{C}))$ be an instrument, \emph{i.e.}, $\mathcal{J}(A)$ is completely positive for every $A\in\mathcal{F}$, $\mathcal{J}$ is $\sigma$-additive and $\mathcal{J}(\Omega)$ is a quantum channel. We say that $\mathcal{J}$ is informationally complete if there exist $V:\mathbb{C}^d\to\mathbb{C}^d\otimes\mathbb{C}^k$, $\mu$ a measure over $(\Omega,\mathcal{F})$ and $M:\Omega\to M_{k,\geq}(\mathbb{C})$ an element of $L^1(\mu)$ such that
		\begin{align}\label{eq:IC POVM}
			\left\{\int_\Omega f(\omega) M(\omega) \d \mu(\omega): f:\Omega\to \mathbb{C} \textnormal{ bounded and measurable}\right\} = M_k(\mathbb{C})
		\end{align}
		and 
		\begin{align}\label{eq:Instrument dilation}
			\mathcal{J}(A): X \mapsto \int_A V^*(X\otimes M(\omega)) V \d \mu(\omega)
		\end{align}
	for all $A\in\mathcal{F}$.  With a small abuse of notation we denote
	$$\mathcal J(\omega):X\mapsto V^*(X\otimes M(\omega))V.$$
	\end{definition}
Without assuming \Cref{eq:IC POVM}, \cite[Theorem~7.11]{busch2016quantum} yields that a dilation as in \Cref{eq:Instrument dilation} exists for any instrument $\mathcal J$. This is a consequence of Stinespring's dilation theorem.

Then, \Cref{eq:IC POVM} is actually the condition for $M$ being an informationally complete positive operator-valued measure (IC-POVM). Equivalently, IC-POVMs are POVMs such that the map $\rho\mapsto \tr(\rho M(\omega))\d\mu(\omega)$ from the set of density matrices to probability measures over $\Omega$ is injective. Hence, knowing the distribution of outcomes of the measurement modeled by the POVM allows for the identification of the system state. So, an IC-POVM can be used to perform a complete tomography of the system state.

Considering that instruments can always be thought of as resulting from an indirect measurement via probes, see \cite[Theorem~7.14]{busch2016quantum}, informationally complete instruments correspond to indirect measurements allowing for a complete probe tomography. \Cref{lem:IC instrument space} shows that for informationally complete instruments, in contrast to the probe POVM $M$, the induced POVM on the system, $M_p(\omega_1,\dotsc,\omega_p):=\mathcal J(\omega_1)\circ\dotsb\circ\mathcal J(\omega_p)(\id_{\cc^d})$ with arbitrarily large $p$, is not always informationally complete. It is so when $\Phi=\mathcal J(\Omega)$ is primitive and $p$ is large enough, since $m$ in \Cref{lem:IC instrument space} is actually the period of $\Phi$ and $m=1$ for primitive maps.

The notion of informationally complete instruments has been used in \cite{BenCuneoQDB} to prove the equivalence between the vanishing of entropy production of repeated quantum measurements and a notion of quantum detailed balance.

\begin{lemma}
    \label{lem:IC implies pur}
    Assume $\Phi$ is irreducible and $\mathcal J$ is an informationally complete instrument, see \Cref{def:IC}, such that $\mathcal J(A)(X)=\int_A v(\omega)^*Xv(\omega)\d\mu(\omega)$ for some measurable $v:\Omega\to M_d(\cc)$. Then, $\mu$ purifies in the sense of \Cref{def:purification}.
\end{lemma}
\begin{proof}
Following \Cref{lem:IC instrument space}, there exist an orthogonal partition $\{E_i\}_{i=1}^m$ of $\cc^d$ and $p\in \nn$ such that
$$\left\{\int_{\Omega^p}f(\omega_1,\dotsc,\omega_p) \mathcal J(\omega_1)\circ\dotsb\circ\mathcal J(\omega_p)(\id_{\cc^d})\d\mu^{\otimes p}(\omega_1,\dotsc,\omega_p): f\in L^{\infty}(\mu^{\otimes p})\right\}=\bigoplus_{i=1}^m\mathcal L(E_i)$$
where for any finite dimensional linear space $E$, $\mathcal L(E)$ denotes the set of endomorphisms of $E$.
Then, purification is equivalent to: Any orthogonal projector $\pi$ such that 
\begin{equation} \label{eq:pur_blockdecomposition}
    \pi\left(\bigoplus_{i=1}^m\mathcal L(E_i)\right)\pi=\cc\pi
\end{equation}
is of rank at most $1$.

So let $\pi$ to be a nonzero orthogonal projector satisfying \Cref{eq:pur_blockdecomposition}. Let $\{p_1,\dotsc,p_d\}$ be a resolution of the identity by rank-$1$ orthogonal projectors such that $p_j\in \bigoplus_{i=1}^m\mathcal L(E_i)$ for any $j\in \{1,\dotsc,d\}$. Since the $p_j$ are positive semi-definite, and sum up to the identity, there exists $j\in \{1,\dotsc,d\}$ such that $\pi p_j\pi\neq 0$. The rank submultiplicativity implies that $\pi p_j \pi$ has rank $1$. Then, from \Cref{eq:pur_blockdecomposition}, $\pi p_j\pi=\|p_j\pi\|^2\pi$, so also $\pi$ has rank $1$. Hence, purification holds, and the lemma is proved.
\end{proof}

\begin{proof}[Proof of \Cref{thm:purification}]
We will apply \Cref{lem:IC implies pur}, which readily establishes the claimed purification assumption. For this, set
$\Omega=U(k)$ and $\mathcal{F}=\mathcal{B}(U(k))$, \emph{i.e.}, $\mathcal{F}$ is the Borel $\sigma$-algebra over the unitary group $U(k)$. Recall that $\Phi:X\mapsto \sum_{i=1}^k v_i^* X v_i$ for some 
$(v_i)_{i=1}^k \in M_d(\cc)^k$ and consider as an instrument $\mathcal J(u):X\mapsto v(u)^*Xv(u)$, where we recall $v(u) = \sum_{j=1}^k u_{1j} v_j$. By \Cref{def:IC}, we may apply \Cref{lem:IC implies pur} and therefore immediately conclude the proof, if we can find some $(M,V)$ such that \Cref{eq:IC POVM,eq:Instrument dilation} hold.\\
Let $(e_i)_{i=1}^k$ be the canonical orthonormal basis of $\cc^k$ and let $V:=\sum_{i=1}^k v_i\otimes e_i$. Then, setting $M:U(k)\to M_d(\cc)$, $M(u):=u^*|e_1\rangle\langle e_1|u$, we can write
$$\mathcal J(u):X\mapsto V^*(X\otimes M(u))V ,$$
and $\mathcal J (A) = \int_A \mathcal J(u) \d \mu(u)$ implies \Cref{eq:Instrument dilation}. So it remains to show that the POVM $M$ verifies \Cref{eq:IC POVM} whenever $\mu$ is non-singular.
Let $B:=\{\int_{U(k)} f(u)M(u)\d\mu(u) : f\in L^{\infty}(\mu)\}$.
By definition and continuity of $M$ in $u$,
$$B=\operatorname{linspan}\{u^*|e_1\rangle\langle e_1|u : u\in \supp \mu\}.$$
Assume $B\neq M_k(\cc)$. Then, there exists a nonzero $X\in M_k(\cc)$ such that $\tr(u^*|e_1\rangle\langle e_1|u X)=0$ for any $u\in \supp \mu$. Then, \cite[Proposition~1]{mityagin_zero_2020} applied to the polynomial $u\mapsto \tr(u^*|e_1\rangle\langle e_1|u X)$ implies $\mu_\unif(\supp \mu)=0$. That implies $\mu$ does not verify the non-singularity assumption of \Cref{def:non singular}. This contradiction yields the theorem.
\end{proof}

\subsection{Regularity: Proof of \Cref{thm:phi-irred,thm:absolute continuity}}\label{sec:proof irred}
We start with the proof of \Cref{thm:phi-irred}.
Recalling \Cref{def:phi irr} and \Cref{def:mPrim}, we must prove that for any $\hat x\in \pcd$ and $A \subseteq \pcd$ measurable, such that $\nu_\unif(A)>0$, we have
\begin{align}
    \label{eq:proof phi irr}
    \Pi^p(\hat x,A)>0 ,
\end{align}
with $p$ being the $\hat x$-dependent integer of \Cref{def:mPrim}. Since $\nu_\unif(A)>0$ only if $A$ has non-empty interior, it is sufficient to prove \Cref{eq:proof phi irr} for $A$ being a ball $B(\hat y,\varepsilon)$, with respect to the metric $d(\hat x, \hat y)$ defined above \Cref{thm:uniqueness invariant}, centered in an arbitrary $\hat y\in \pcd$ and with arbitrary radius $\varepsilon>0$.

Actually, \Cref{eq:proof phi irr} is the only statement that needs to be proved, since then $\nu_\inv\gg\nu_\unif$ is implied by the following proposition from \cite{MT}.
\begin{proposition}[{\cite[Proposition~10.1.2]{MT}}, Item (ii)] \label{prop:phi_irreducibility_ac}
Let $ (\hat{x}_n)_{n \in \nn} $ be a $ \varphi $-irreducible Markov chain with $ \varphi = \nu_\unif $ on $ \pcd $. Then $ \nu_\inv \gg \nu_\unif $.
\end{proposition}

So let us now prove \Cref{eq:proof phi irr}. Fix $\hat x,\hat y\in \pcd$ and $\varepsilon>0$. Let $p\in \nn$ be the $\hat x$-dependent integer of \Cref{def:mPrim}. Let $\{v_i\}_{i=1}^k$ be a set of matrices such that $\Phi:X\mapsto \sum_{i=1}^kv_i^* X v_i$. 

The assumption of multiplicative primitivity (\emph{i.e.} \Cref{def:mPrim}) implies that there exist $\lambda>0$ and $(\phi_1,\dotsc,\phi_p)\in S(\cc^k)^{p}$ such that
$$y=\lambda\sum_{i_1,\dotsc, i_p=1}^k (\langle e_{i_1},\phi_1\rangle v_{i_1})\dotsb(\langle e_{i_p},\phi_p\rangle v_{i_p}) x.$$
Choosing $(u(1),\dotsc,u(p))\in U(k)^p$ such that $u(l)_{1i}=\langle e_{i},\phi_l\rangle$ for any $l\in \{1,\dotsc,p\}$ and $i\in\{1,\dotsc,k\}$,
$$\hat y=v(u(p))\dotsb v(u(1))\cdot \hat x.$$
Then,
\begin{align*}
    &\Pi^p(\hat x, B(\hat y,\varepsilon))=\int_{U(k)^p}\chi(d(v(w_p)\dotsb v(w_1)\cdot \hat x,\hat y)<\varepsilon)\|v(w_p)\dotsb v(w_1)x\|^2\d\mu^{\otimes p}(w_1,\dotsc,w_p)\\
            &=\int_{U(k)^p}\chi(d(v(w_p)\dotsb v(w_1)\cdot \hat x,v(u(p))\dotsb v(u(1))\cdot \hat x)<\varepsilon)\|v(w_p)\dotsb v(w_1)x\|^2 \d\mu^{\otimes p}(w_1,\dotsc,w_p) ,
\end{align*}
with $\chi(\mathsf{S})=1$ if the statement $\mathsf{S}$ is true, and $\chi(\mathsf{S})=0$ otherwise.

Since $v(u(p))\dotsb v(u(1))x\neq 0$, $(w_1,\dotsc,w_p)\mapsto v(w_p)\dotsb v(w_1)\cdot \hat x$ is continuous in an open neighborhood of $(u(1),\dotsc,u(p))$. Hence, there exists $\eta>0$ such that $(w_1,\dotsc,w_p)\in B(u(1),\eta)\times\dotsb\times B(u(p),\eta)$ implies $d(v(w_p)\dotsb v(w_1)\cdot \hat x,v(u(p))\dotsb v(u(1))\cdot \hat x)<\varepsilon$. So,
\begin{align*}
    \Pi^p(\hat x, B(\hat y,\varepsilon))\geq &\int_{B(u(1),\eta)\times\dotsb\times B(u(p),\eta)}\|v(w_p)\dotsb v(w_1)x\|^2\d\mu^{\otimes p}(w_1,\dotsc,w_p).
\end{align*}
Again, since $(w_1,\dotsc,w_p)\mapsto \|v(w_p)\dotsb v(w_1)x\|^2$ is continuous, for $\eta>0$ small enough, there exists $\delta>0$ such that
$$\inf_{(w_1,\dotsc,w_p)\in B(u(1),\eta)\times\dotsb\times B(u(p),\eta)} \|v(w_p)\dotsb v(w_1)x\|^2\geq \delta.$$
It follows that
$$\Pi^p(\hat x, B(\hat y,\varepsilon))\geq \delta\prod_{l=1}^p\mu(B(u(l),\eta)).$$
Since $\mu\gg \mu_\unif$, for any $l\in\{1,\dotsc,p\}$, $\mu(B(u(l),\eta))>0$ and \Cref{thm:phi-irred} is proved.\hfill\qed

\bigskip
We turn to the proof of \Cref{thm:absolute continuity}. The proof relies on the fact that when $\mu$-almost all of the matrices $v(u)$ are invertible, $\nu_\inv$ is either pure point, absolutely continuous with respect to $\nu_\unif$, or singular continuous. It cannot be a mixture of the three. A similar result for products of independent and identically distributed matrices can be found in \cite[Proposition~4.4, Section~VI.4]{BL85}. We follow essentially the proof in that reference. But, before, we prove the lemma justifying \Cref{rem:invertible everywhere}. It shows that when $\mu\ll\mu_\unif$, it is sufficient to check that $v(u)$ is invertible for one $u\in U(k)$.
\begin{lemma}
    \label{lem:invertible evrywhere}
    Assume $\mu\ll\mu_\unif$ and there exists $u\in U(k)$ such that $v(u)$ is invertible. Then, $v(u)$ is invertible for $\mu$-almost every $u$.
\end{lemma}
\begin{proof}
    Assume there exists a measurable subset $A$ of $U(k)$ such that $\mu(A)>0$ and $v(u)$ is non-invertible for any $u\in A$. Assuming $v(u)$ is non-invertible for any $u\in A$ means that for any $u\in A$, $\det v(u)=0$. Since there exists $u\in U(k)$ such that $\det v(u)\neq 0$, \cite[Proposition~1]{mityagin_zero_2020} implies $\mu_\unif(A)=0$ and therefore $\mu(A)=0$ since $\mu\ll\mu_\unif$. This contradiction yields the theorem.
\end{proof}

We now turn to the proof of our version of \cite[Proposition~4.4, Section~VI.4]{BL85}.
\begin{lemma}
    \label{lem:nu_inv_regularity}
    Assume that for $\mu$-almost every $u$, $v(u)$ is invertible. Assume $\Pi$ accepts a unique invariant probability measure $\nu_\inv$. Then, one of the three following alternatives occur:
    \begin{itemize}
        \item $\nu_\inv$ is pure point,
        \item $\nu_\inv\ll\nu_\unif$,
        \item $\nu_\inv\perp\nu_\unif$ and has no pure point component. 
    \end{itemize}
\end{lemma}
\begin{proof}
    Let $\nu_{ac}$ be the absolutely continuous part of $\nu_\inv$ with respect to $\nu_\unif$, $\nu_{sc}$ its singular continuous part and $\nu_{pp}$ its pure point part in Lebesgue's decomposition of Borel measures. Denote furthermore $\nu_c=\nu_{ac}+\nu_{sc}$. First, there is no $\hat y\in \pcd$ such that $\int_\pcd\Pi(\hat x,\{\hat y\})\d\nu_c(\hat x)>0$. Indeed, if it was the case, there would exist a measurable subset $A$ of $U(k)\times\pcd$ such that $(\mu\otimes\nu_c)(A)>0$ and for any $(u,\hat x)\in A$, $v(u)$ is invertible and $v(u)\cdot\hat x=\hat y$. That is equivalent to $\hat x=v(u)^{-1}\cdot\hat y$ for any $(u,\hat x)\in A$. Hence, $(\mu\otimes\nu_c)(A)\leq \int_{U(k)} \nu_c(\{\hat x: \hat x=v(u)^{-1}\cdot \hat y\})\d\mu(u)=0$ since $\{\hat x: \hat x=v(u)^{-1}\cdot \hat y\}$ is at most a singleton for $\mu$-almost every $u$. That contradicts the assumption $(\mu\otimes\nu_c)(A)>0$. Thus,
    $\nu_c\Pi$ has no pure point component. Then, with
    $$\nu_\inv=\nu_{pp}+\nu_c=\nu_{pp}\Pi+\nu_c\Pi$$
    we conclude $\nu_c\geq \nu_c\Pi$. Since $\Pi\one(\hat x)=1$, $\nu_c\Pi(\pcd)=\nu_c(\pcd)$ and it follows that $\nu_c=\nu_c\Pi$. By uniqueness of the invariant measure, either $\nu_c=0$ and therefore $\nu_\inv=\nu_{pp}$, or $\nu_c(\pcd)=1$ and $\nu_\inv=\nu_c$.
    
    It remains to show that for $\nu_c(\pcd)=1$, either $\nu_{\inv} = \nu_{ac}$ or $\nu_{\inv} = \nu_{sc}$. For any $v$ invertible, the measure defined by $\nu_v(A)=\int_\pcd \one_A(v\cdot\hat x)\|vx\|^2\d\nu_\unif(\hat x)$ is equivalent to $\nu_\unif$. Indeed, for any measurable subset $A$ of $\pcd$,
    $$\|v^{-1}\|^{-2}\int_\pcd\one_{A}(v\cdot \hat x)\d\nu_\unif(\hat x)\leq \nu_v(A)\leq \|v\|^2\int_\pcd\one_{A}(v\cdot \hat x)\d\nu_\unif(\hat x).$$
    Since the image measure of $\nu_\unif$ by $\hat x\mapsto v\cdot \hat x$ is equivalent to $\nu_\unif$ (see \cite[Exercise~5.4, Section~I.5]{BL85} for example), $\nu_v$ is equivalent to $\nu_\unif$. Since $\nu_\unif\Pi=\int_{U(k)}\nu_{v(u)}\d\mu(u)$ and $v(u)$ is invertible for $\mu$-almost every $u$, $\nu_\unif\Pi$ is equivalent to $\nu_\unif$. Hence, if $\nu_\unif(A)=0$, then $\nu_\unif\Pi(A)=0$ and $\nu_{ac}(A)=0$. Assume $\nu_{ac}\Pi(A)>0$. Then there exists a measurable $B\subset \pcd$ such that $\nu_{ac}(B)>0$ and $\forall \hat x \in B$, $\Pi(\hat x,A)>0$. The absolute continuity $\nu_{ac}\ll\nu_\unif$ implies $\nu_\unif(B)>0$. Then, $\Pi(\hat x,A)>0$ for any $\hat x\in B$ implies $\nu_\unif\Pi(A)>0$ which contradicts $\nu_\unif\Pi(A)=\nu_\unif(A)=0$. Therefore, $\nu_{ac}\Pi\ll\nu_\unif$, and
    $$\nu_c=\nu_{ac}+\nu_{sc}=\nu_{ac}\Pi+\nu_{sc}\Pi$$
    implies $\nu_{ac}\geq \nu_{ac}\Pi$. Thus, we infer as before that either $\nu_\inv=\nu_{ac}$ or $\nu_\inv=\nu_{sc}$. This concludes the proof of the lemma.
\end{proof}

\begin{proof}[Proof of \Cref{thm:absolute continuity}]
    Since $v(u)$ is invertible for $\mu$-almost any $u$, \Cref{lem:nu_inv_regularity} implies that either $\nu_\inv$ is pure point, absolutely continuous or singular continuous with respect to $\nu_{\unif}$. Since $\mu\gg\mu_\unif$, \Cref{thm:phi-irred} implies $\nu_\inv\gg \nu_\unif$, therefore $\nu_\inv$ cannot be pure point or singular continuous. Hence, $\nu_\inv\sim\nu_\unif$.
\end{proof}

\subsection{Symmetries: Proof of \Cref{thm:symmetries}}\label{sec:proof sym}
The proof of this theorem and the related corollary rely on the observation that when $\mu$ is uniform, then $\delta_{\hat x}\Pi$ is the GAP measure associated to $\Phi^*(|x\rangle\langle x|)$. We refer the reader to Appendix~\ref{app:GAP} and \cite{JRW94,GLTZ06,GLMTZ16,Tum20} for an introduction to GAP measures.

In the appendix and those references, the GAP measure is defined on $S(\cc^d)$ and not $\pcd$. Given a GAP measure on $S(\cc^d)$, one can always define the respective GAP measure on $\pcd$ as a pullback under the projection $S(\cc^d) \to \pcd$. Conversely, by Item~(c) in \Cref{prop:GAP}, GAP measures (on $S(\cc^d)$) are invariant under multiplication by a phase, so given a GAP measure on $\pcd$, one can uniquely reconstruct the respective GAP measure on $S(\cc^d)$. Thus, the notions of GAP measures on $S(\cc^d)$ and $\pcd$ are equivalent.

\begin{lemma}
    \label{lem:GAP}
    Assume $\mu$ = $\mu_\unif$. Then $\delta_{\hat x}\Pi$ is the GAP measure associated to $\Phi^*(|x\rangle\langle x|)$.
\end{lemma}
\begin{proof}
    This lemma is a consequence of \cite[Lemma~1]{GLMTZ16}: Let $v_1,\dots,v_k$ be Kraus operators of $\Phi$. We define $V:\mathbb{C}^d\to\mathbb{C}^d\otimes \mathbb{C}^k$ by
    \begin{align*}
        |Vx\rangle = \sum_j v_j |x\rangle \otimes |e_j\rangle,
    \end{align*}
    where  $(e_j)$ is the canonical basis of $\mathbb{C}^k$. Then,
    \begin{align*}
        \Phi^*(|x\rangle \langle x|) = \tr_{\cc^k} |Vx\rangle\langle Vx|,
    \end{align*}
    where $\tr_{\cc^k}$ is the partial trace over $\cc^k$ in $\mathbb{C}^d\otimes\mathbb{C}^k$. Let $\Psi := Vx \in S(\mathbb{C}^d\otimes\mathbb{C}^k)$ and $\{u_1,\dots,u_k\}$ be an orthonormal basis of $\cc^k$. Let $U$ be the unitary matrix whose $j^{\rm th}$ row is the entry-wise complex conjugate of $u_j$ for $j$ in $\{1,\dotsc,k\}$. Then, the conditional wave function $\psi^U$ is the random $\cc^d$-vector defined as
    \begin{align*}
        \psi^U = \frac{(\id_{\cc^d}\otimes\langle u_J|)|Vx\rangle}{\|(\id_{\cc^d}\otimes\langle u_J|)|Vx\rangle\|},
    \end{align*}
    where $J$ is random with distribution $\mathbb{P}(J=j) = \|(\id_{\cc^d}\otimes\langle u_J|)|Vx\rangle\|^2$. Direct computation leads to $(\id_{\cc^d}\otimes\langle u_j|)|Vx\rangle =v_j(U) x$ and therefore $\psi^U = \frac{v_J(U)x}{\|v_J(U)x\|}$. Let $\mu_1^{\Psi,U}$ denote the distribution of the conditional wave function $\psi^U$. Then, for any measurable set $A\subset S(\mathbb{C}^d)$, we find that
    \begin{align*}
    \mu_1^{\Psi,U}(A) = \sum_j \|v_j(U)x\|^2 \mathbbm{1}_A(v_j(U)\cdot \hat{x}).
    \end{align*}
    Now let $U$ be distributed according to the Haar probability measure $\lambda$ over $U(k)$. Then, the random variables $(v_j(U))_{j=1}^k$ are identically distributed and
    \begin{align*}
        \ee_\lambda(\mu_1^{\Psi,U}(A))= k\ee_\lambda(\|v(U)x\|^2 \mathbbm{1}_A(v(U)\cdot \hat{x})) = \delta_{\hat{x}}\Pi(A),
    \end{align*}
    where we used that $\ee_\lambda(v(U)^*X v(U))=\frac1k\Phi(X)$.
    From \cite[Lemma~1]{GLMTZ16} it follows that
    \begin{align*}
        \mathbb{E}_\lambda \mu_1^{\Psi,U} = \GAP_{\tr_{\cc^k} |\Psi\rangle\langle\Psi|} = \GAP_{\Phi^*(|x\rangle\langle x|)}
    \end{align*}
    and the claim is proved.
\end{proof}

The second and last lemma we require show the impact of the symmetry on the law of the Markov chain itself.
\begin{lemma}
    \label{lem:sym to measure}
    Assume $\mu=\mu_\unif$. Assume $U\in U(d)$ is such that for any $X\in \Ld$, $\Phi(U^*XU)=U^*\Phi(X)U$. Then, for any $\hat x\in \pcd$, $\delta_{U\cdot \hat x}\Pi$ is the image measure of $\delta_{\hat x}\Pi$ by the map $\hat y\mapsto U\cdot \hat y$.
\end{lemma}
\begin{proof}
    Property 2 of GAP measures in \cite{GLTZ06} implies that GAP of $U\Phi^*(|x\rangle\langle x|)U^*$ is equal to the image measure of GAP of $\Phi^*(|x\rangle\langle x|)$ by the map $\hat y\mapsto U\cdot\hat y$. Then the symmetry $\Phi^*(U|x\rangle\langle x|U^*)=U\Phi^*(|x\rangle\langle x|)U^*$ and \Cref{lem:GAP} yield the lemma.
\end{proof}
\begin{proof}[Proof of \Cref{thm:symmetries}]
    Since $\mu=\mu_\unif$, it is trivially non-singular and \Cref{thm:uniqueness invariant} yields that $\Pi$ accepts a unique invariant measure $\nu_\inv$. Let $\nu_U$ be the image measure of $\nu_\inv$ by $\hat x\mapsto U\cdot \hat x$. Then, $\nu_U\Pi=\int_\pcd \delta_{U\cdot\hat x}\Pi \d\nu_\inv(\hat x)$. Then, \Cref{lem:sym to measure} implies that $\nu_U\Pi$ is the image measure of
    $\nu_\inv\Pi$ by $\hat x\mapsto U\cdot \hat x$. Since $\nu_\inv\Pi=\nu_\inv$, we deduce that $\nu_U$ is $\Pi$-invariant and the uniqueness of the invariant measure yields the theorem.
\end{proof}

\section{Examples}\label{sec:xpl}

\subsection{Invariant measure non-equivalent to uniform}\label{sec:xpl non equivalent uniform}
Let us provide a counterexample to \Cref{thm:absolute continuity} where $\Phi$ is primitive, $\mu\gg \mu_\unif$, but $\mu(\{u:\det v(u)=0\})>0$. We show that for this example, $\nu_\inv\gg \nu_\unif$ as implied by \Cref{thm:phi-irred} but there exists a measurable subset $A$ of $\pcd$ such that $\nu_\inv(A)>0$ but $\nu_\unif(A)=0$. 

Let $d=k=2$, $v_1=|e_2\rangle \langle e_1|$ and $v_2=|e_+\rangle \langle e_2|$ with $\{e_1,e_2\}$ the canonical basis of $\cc^2$ and $\hat e_+$ the equivalence class of $e_1+e_2$. Now $\Phi^3(X)=\frac14X_{22}(\id_{\cc^2}+|e_1\rangle\langle e_1|) + X_{++} (\frac12 \id_{\cc^2}+\frac14|e_2\rangle\langle e_2|)$ with $X_{22}=\langle e_2,Xe_2\rangle$ and $X_{++}=\langle e_+,Xe_+\rangle$. Since $|e_2\rangle\langle e_2|+|e_+\rangle\langle e_+|$ is positive definite, we conclude $\Phi^3(X)>0$ for any $X\geq 0$ and by \Cref{prop:mPrim dim 2}, $\Phi$ is multiplicatively primitive. Let\footnote{At this point, we normalize $\mathrm{Haar}_{U(2)}(U(2)) = 1$.} $\lambda=\frac12{\rm Haar}_{U(2)}+\frac12 \delta_{\id_{\cc^2}}$ and define $\mu$ as in \Cref{eq:def mu}. Then, $\mu\gg \mu_\unif$ and for any probability measure $\nu$ over $\pc{2}$,
$$\nu\Pi(\{\hat e_2,\hat e_+\})\geq \int_{\pc{2}}  \frac12 (|\langle e_1,x\rangle|^2+|\langle e_2,x\rangle|^2) \d\nu(\hat x)=\frac12>0.$$
Thus, $\nu_\inv(\{\hat e_2,\hat e_+\})>0$, but $\nu_\unif(\{\hat e_2,\hat e_+\})=0$.

\subsection{Uniformly randomized quantum trajectories with explicit $\nu_{\inv}$}
In this section, we always consider $\mu=\mu_\unif$. 

As a simple example with an explicit $\nu_{\inv}$, consider the projection channel $\Phi(X)=\id_{\cc^d}\tr(\rho X)$ for some positive definite density matrix $\rho$. So $\Phi^*(\varrho) = \rho$ for any density matrix $\varrho$. Then $\Phi$ is trivially positivity improving and therefore multiplicatively primitive thanks to \Cref{prop:pos improve implies mPrim}. Moreover, \Cref{lem:GAP} implies $\nu_\inv=\GAP_\rho$ where $\GAP_\rho$ is the GAP measure for the density matrix $\rho$.\\
Another example with an explicit $\nu_{\inv}$ is the depolarizing channel $\Phi(X)=(1-p)X+p\id_{\cc^d}\frac1d\tr(X)$ for $p\in (0,1]$. Here, the symmetry group of $\Phi$ is the full unitary group and \Cref{cor:full symmetry} implies $\nu_\inv=\nu_\unif$.\\
However, for a generic multiplicatively primitive $\Phi$, and for $\mu=\mu_\unif$, computing a closed form for invariant measures does not seem to be an easy task. For instance, as soon as the maximally mixed state $\id_{\cc^d}/d$ in the depolarizing channel is replaced by an arbitrary positive definite density matrix, establishing a closed form for $\nu_{\inv}$ is nontrivial, as illustrated in the next section.

\subsection{Uniformly randomized in dimension $2$}
Let us now concentrate on $d=2$. In that case our different assumptions reduce to standard ones and the density of $\nu_\inv$ with respect to $\nu_\unif$ satisfies an integral kernel equation.
\begin{proposition}\label{prop:2D}
Assume $d=2$ and $\Phi$ is primitive, see \Cref{def:prim}. Let $\mu=\mu_\unif$. Then, $\nu_\inv\sim\nu_\unif$ and the density $f_\inv(\hat x)=\frac{\d\nu_\inv}{\d\nu_\unif}(\hat x)$ is a solution to 
\begin{align}\label{eq:fix point f}
f_\inv(\hat x)=\int_{\pc{2}} \frac{2}{\det\Phi^*(|x'\rangle\langle x'|)}\langle x,\Phi^*(|x'\rangle\langle x'|)^{-1}x\rangle^{-3}f_\inv(\hat x')\d\nu_\unif(\hat x').
\end{align}
\end{proposition}

This result is an application of \Cref{lem:GAP} and \cite[Equation~(18)]{GLTZ06} (see Equation~\eqref{eq:density of GAP}). It also requires $\nu_\inv\sim\nu_\unif$ which is assured by primitivity, \Cref{prop:mPrim dim 2}, \Cref{thm:absolute continuity} and \Cref{lem:2D inversible}. 

\begin{remark}
    If $\Phi$ is irreducible but not primitive (\emph{i.e.} its period is $2$), then there exists an orthonormal basis $\{x_1,x_2\}$ of $\cc^2$ such that $\Phi^*(|x_1\rangle\langle x_1|)=|x_2\rangle \langle x_2|$ and $\Phi^*(|x_2\rangle\langle x_2|)=|x_1\rangle \langle x_1|$. Thus, if $\Pi$ accepts a unique invariant probability measure, it is equal to $\frac12(\delta_{\hat x_1}+\delta_{\hat x_2})$.
\end{remark}
\begin{remark}
    With minor modifications to the proof of \Cref{prop:2D}, \Cref{eq:fix point f} can be generalized to arbitrary dimension $d$ using \cite[Equation~(18)]{GLTZ06}, provided $\det v(u)\neq 0$ for almost all $u\in U(k)$, so we can apply \Cref{thm:absolute continuity}. Thanks to \Cref{lem:invertible evrywhere}, it even suffices to have $\det v(u)\neq 0$ for only one single $u\in U(k)$.
\end{remark}

In $d=2$, the required invertibility of $v(u)$ is actually automatically true for irreducible $\Phi$.

\begin{lemma}
    \label{lem:2D inversible}
    Assume $d=2$ and $\Phi$ is irreducible, see \Cref{def:irr}. Then, $\det v(u)\neq 0$ for $\mu_\unif$-almost every $u\in U(k)$.
\end{lemma}
\begin{proof}
    Assume $\det v(u) = 0$ within some non-null set with respect to $\mu_\unif$. Then, there is some open set $\mathcal{U} \subset U(k)$ with $\det v(u) = 0 \; \forall u \in \mathcal{U}$.
    Now, recall that $\det v(u)$ is a polynomial in the matrix entries of $u$, so $\det v(u) = 0 \; \forall u \in U(k)$. Fix $u_1, u_2\in U(k)$ such that $v(u_1)\not\propto v(u_2)$ and $v(u_1)$ and $v(u_2)$ are non zero.
    
    Consider first the case where $v(u_1)$ is diagonalizable. Then, for any $w\in \cc$, there exists $u_w\in U(k)$ such that $v(u_w)\propto wv(u_1)+v(u_2)$. Working in an eigenbasis of $v(u_1)$, we can write any $v(u_w)$, as a scalar multiple of
    $$X_w:= w v(u_1)+ v(u_2) = \begin{pmatrix} x+\alpha w & y\\z & t\end{pmatrix} \quad\mbox{where } v(u_2)=\begin{pmatrix} x & y\\z & t\end{pmatrix} \in M_2(\cc)$$
    where $\alpha$ is the non zero eigenvalue of $v(u_1)$.
    Since $v(u_w)$ is not invertible for any $w\in \cc$, $0 = \det X_w=(xt-yz)+w\alpha t$ for all $w\in \cc$, so $t=yz=0$. Since $u_2$ was arbitrary, for any $u \in U(k)$, we have the dichotomy:
    $$\text{either}\quad v(u)=\begin{pmatrix} a(u)& b(u)\\0&0\end{pmatrix}\quad \text{or}\quad v(u)=\begin{pmatrix} a(u)& 0\\c(u)&0\end{pmatrix},$$
    for some $a,b,c: U(k) \to \cc$. Moreover, either $b \equiv 0$ or $c \equiv 0$, since if we had $u_3, u_4 \in U(k)$ with $b(u_3)c(u_4)\neq 0$, then there would be some $u \in U(k)$ with $v(u)=\frac{1}{\sqrt{2}}(v(u_3)+v(u_4))$ not being in one of the above forms.
    
    Now, let $f$ be an non-zero element of $\ker v_1$.
    If $b \equiv 0$, then $v(u) f=0$ for any $u\in U(k)$, thus $\int_{U(k)}\|v(u)f\|^2\d\mu(u)=0$ which is incompatible with the condition $\int_{U(k)}v(u)^*v(u)\d\mu(u)=\id_{\cc^2}$. Thus, $b\not\equiv 0$ and $c \equiv 0$. Then, the image of $v(u_1)$ is a one-dimensional invariant subspace of $v(u)$ for any $u\in U(k)$. That contradicts the irreducibility of $\Phi$. Therefore $v(u)$ is invertible for $\mu$-almost every $u$.
   
    Finally, if $v(u_1)$ is not diagonalizable, working in a basis of its Jordan form,
    $$X_w=wv(u_1)+v(u_2)=\begin{pmatrix} x & y+w\\z & t\end{pmatrix},$$
    and $\det X_w=xt-yz-zw$. Hence, $\det X_w=0$ for any $w\in \cc$ if and only if $z=xt=0$. It follows that $v(u)$ is not invertible for any $u\in U(k)$ only if, for any $u\in U(k)$,
    $$v(u)=\begin{pmatrix} a(u)& b(u)\\0&c(u)\end{pmatrix}.$$
    That implies that for any $u\in U(k)$, the image of $v(u_1)$ is an invariant subspace of $v(u)$ for all $u \in U(k)$, which contradicts again the irreducibility of $\Phi$.
\end{proof}
\begin{proof}[Proof of \Cref{prop:2D}]
    Since $d=2$ and $\Phi$ is primitive, \Cref{prop:mPrim dim 2} implies that $\Phi$ is multiplicatively primitive. Then, \Cref{lem:2D inversible} and \Cref{thm:absolute continuity} imply $\nu_\inv\sim\nu_\unif$. Now assume $\Phi^*(|y\rangle\langle y|)$ is invertible for $\nu_\unif$-almost every $\hat y$. Then \Cref{lem:GAP} and \eqref{eq:density of GAP} yield the proposition.
    
    It remains to prove that $\Phi^*(|y\rangle\langle y|)$ is invertible for $\nu_\unif$-almost every $\hat y$. If $\Phi^*(|y\rangle\langle y|)$ is invertible for one $\hat y\in \pc{2}$, then $\det \Phi^*(|y\rangle\langle y|)\neq0$ and \cite[Proposition~1]{mityagin_zero_2020} imply that $\Phi^*(|y\rangle\langle y|)$ is invertible for $\nu_\unif$-almost every $\hat y$.
    
    So we must exclude that $\Phi^*(|y\rangle\langle y|)$ is non-invertible for all $\hat y\in \pc{2}$. If this was the case, then for any $\hat y\in \pc{2}$, there would exist a sequence $(\hat y_n)_{n\in \nn}$ in $\pc{2}$ such that ${\Phi^*}^n(|y\rangle \langle y|)=|y_n\rangle \langle y_n|$ for any $n\in \nn$. That contradicts the assumption that $\Phi$ is primitive.
\end{proof}

\begin{remark}
    Already for the simple channel $\Phi(X)=(1-p)X+p\id_{\cc^2}\tr(\rho X)$ with $p\in (0,1)$ and $\rho$ a positive definite density matrix distinct from $\id_{\cc^2}/2$, solving \Cref{eq:fix point f} for $f_\inv$ is non-trivial. 
    
    During the production of this article a result showing stability of the GAP measure under measurement appeared \cite{tumulka2026gap}. However, this stability concerns the state conditioned on the measurement outcome (\emph{i.e.} for a fixed predetermined outcome) and as mentioned in the same reference that does not imply invariance of the GAP measure under indirect measurement.
\end{remark}

\subsection{Examples in dimension $3$ satisfying multiplicative primitivity}\label{sec:xpl 3D}
Recall for a quantum channel $\Phi$ that ``positivity improving'' implies ``multiplicatively primitive'', which implies ``primitive''.
We provide two non-trivial examples for $d=3$ such that multiplicative primitivity holds, but $\Phi$ is not positivity-improving. The method of proof differs between the two examples, but always relies on algebraic independence. In the first example, $v(u)$ is $\mu$-almost always invertible, whereas, in the second one, $v(u)$ is not invertible for any $u$.
\subsubsection{An example with almost sure invertible $v(u)$} We provide a nontrivial example where all our assumptions are verified.
Let,
$$v_1=\frac{\sqrt{2}}{2}\begin{pmatrix}0&1&0\\ 1&0&1\\0&0&0\end{pmatrix}\quad\text{and}\quad v_2=\frac{\sqrt{2}}{2}\begin{pmatrix}0&0&0\\ 0&-1&0\\1&0&-1\end{pmatrix}.$$
\begin{proposition}\label{prop:xpl 1 3D}
The channel $\Phi(X) := v_1^* X v_1 + v_2^* X v_2$ is not positivity-improving, but multiplicatively primitive, see \Cref{def:mPrim}, and for $\mu_\unif$-almost any $u\in U(2)$, $\det v(u)\neq0$.
\end{proposition}
\begin{proof}
    One easily checks $\Phi(|e_1\rangle \langle e_1|) = \frac12|e_2 \rangle \langle e_2|$, which is not positive definite, so $\Phi$ is not positivity-improving.
    
    We will next show that multiplicative primitivity holds with $p=8$. Consider the $3 \times 3$ matrix of polynomials defined by $P(z_1,\dotsc,z_{16})=\prod_{i=0}^7(z_{2i+1} v_1+z_{2i+2} v_2)$. Let $J(z_1,\dotsc,z_{16})$ be the $9 \times 16$ Jacobian matrix of the entries of this matrix with respect to $(z_1,\dotsc,z_{16})$. Then, 
    $$J(1, 2, 3, 1, 2, 3, 1, 2, 3, 1, 2, 3, 1, 2, 3, 1)$$
    has rank $9$. This computation can be made in SageMath 10 using the code available in Appendix~\ref{app:code}. Then, it follows from \cite[Theorem~2.3]{ehrenborg1993apolarity}, that the polynomials defined by the entries of the matrix $P$ are algebraically independent. Hence, $P(\cc^{16})=M_3(\cc)$, which in particular implies multiplicative primitivity.
    
    It remains to prove that $\det v(u)\neq0$ for almost every $u \in U(k)$. This follows from $\det(av_1+bv_2)\neq0$ for any $a,b\in \cc$ such that $ab \neq 0$. Indeed,
    $$av_1+bv_2=\frac{\sqrt{2}}{2}\begin{pmatrix}0&a&0 \\ a&-b&a\\b&0&-b\end{pmatrix},$$
    so $\det(av_1+bv_2)=\frac{\sqrt{2}}{2}a^2b$ and the proposition is proved.
\end{proof}

\subsubsection{An example with non-invertible $v(u)$}\label{sec:xpl 3D non inv}
We conclude our list of examples with one that is multiplicatively primitive but no Kraus decomposition contains an invertible matrix. So \Cref{thm:absolute continuity} cannot be applied. For this, we consider the matrices
$$w_1=\begin{pmatrix} 1&1&0\\ -1&1&0\\0&0&0\end{pmatrix}\quad\mbox{and}\quad w_2=\begin{pmatrix}0&0&1\\ 0&0&1\\ 1&0&0\end{pmatrix},$$
and turn them into Kraus operators $v_1, v_2$ by a similarity transformation. To do so, we first show that
$$T:X\mapsto w_1^* X w_1 +w_2^* X w_2$$
is irreducible. Let us assume $T$ is not irreducible, so from \Cref{thm:irr invariant subsapce} there exists a proper subspace $E \subset \cc^3$ that is both $w_1$- and $w_2$-invariant.
If $e_3\not\perp E$, then $w_1w_2 E = \cc e_1$. Then, $w_1 e_1=e_1-e_2$ and $w_2 e_1=e_3$ imply $E=\cc^3$ and $E$ is not a proper subspace. Hence, $e_3\perp E$ and $E\subset \cc e_1+\cc e_2$. So in that case, $w_2 E=\cc e_3$ or $w_2 E=0$. In the first case, $w_2 E\not\subset E$. In the second case, $E\subset \ker w_2=\cc e_2$. It follows that $E=\cc e_2$, but $w_1 e_2=e_1+e_2$. Thus, there does not exist a proper subspace $E$ of $\cc^3$ such that $w_1E\subset E$ and $w_2E\subset E$.
We conclude that $T$ is irreducible and \cite[Theorem~2.3]{Evans1978} implies that there exist $r>0$ and $C \in M_3(\cc)$, $C>0$ such that $T(C)=\lambda C$. Then, setting
$$v_1=r^{-1/2}C^{1/2}w_1C^{-1/2}\quad \mbox{and}\quad v_2=r^{-1/2}C^{1/2}w_2C^{-1/2},$$
the map $\Phi:M_3(\cc)\to M_3(\cc)$ defined by
$$\Phi(X)=v_1^*X v_1+v_2^* X v_2$$
is a quantum channel.
\begin{proposition}
The channel $\Phi$ just defined is not positivity-improving, but multiplicatively primitive, see \Cref{def:mPrim}, and $\det v(u)=0$ for any $u\in U(2)$.
\end{proposition}
\begin{proof}
    By direct computation, $T(|e_3\rangle \langle e_3|) = |e_1 \rangle \langle e_1|$. We conclude that also $\Phi(|C^{-1/2} e_3\rangle \langle C^{-1/2} e_3|)$ is of rank 1, so $\Phi$ is not positivity-improving.
    
    Next, we verify $\Phi$ is multiplicatively primitive with $p=4$.
    Fix $x\in \cc^3\setminus\{0\}$. Consider the $\cc^3$-vector of polynomials in $(z_1,\dotsc,z_8)$ defined by 
    $$P_x(z_1,\dotsc,z_8)=\left(\prod_{i=0}^3(z_{2i+1} w_1+ z_{2i+2}w_2)\right)x.$$
    To establish multiplicative primitivity, we must show that $P_x: \cc^8\to\cc^3$ is surjective.
    Let $J_x(z_1,\dotsc,z_8) \in M_3(\cc)$ be the Jacobian of $P_x$. Then, for $(z_1,\dotsc,z_8)\in \cc^8$ fixed, the determinant $D_x(z_1,\dotsc,z_8)$ of $(J_x(z_1,\dotsc,z_8)_{ij})_{i,j=1}^3$, is $0$ if and only if $x$ is an element of one of the following three sets
    $$\left\{ \left(a,b,-a\frac{z_2z_6}{z_2z_5+z_1z_6}+\frac{z_7}{z_8}\frac{(a-b)z_2z_5-(a+b)z_1z_6}{z_2z_5+z_1z_6}\right) : a,b\in \cc\right\},$$
    $$\left\{\left(a,b, -2a \frac{z_3z_6}{2z_3z_5+z_4z_6}+\frac{z_7}{z_8}\frac{2(a-b)z_3z_5-(a+b)z_4z_6}{2z_3z_5+z_4z_6}\right) : a,b\in \cc\right\},$$
    or
    $$\left\{\left(a,b, -\frac{a}{2}\frac{z_6}{z_5}-b\frac{z_7}{z_8}\right) : a,b\in \cc\right\}.$$
    These sets can be computed using the SageMath 10 code available in Appendix~\ref{app:code}.
    
    It follows that for any $x\in \cc^3\setminus\{0\}$, there exist $(z_1,\dotsc,z_8)\in \cc^8$ such that $D_x(z_1,\dotsc,z_8)\neq 0$. Then \cite[Theorem~2.3]{ehrenborg1993apolarity} implies that the $3$ entries of $P_x(z_1,\dotsc,z_8)$ are algebraically independent as polynomials in $(z_1,\dotsc,z_8)$. Hence, for any $\hat x\in \pc{3}$,
    $$P_x(\cc^8)=\cc^3.$$
    Since $\left(\prod_{i=0}^3(z_{2i+1} v_1+z_{2i+2} v_2)\right)x=r^{-2}C^{1/2}P((z_1,\dotsc,z_8)\times C^{-1/2}x)$ and $C^{1/2} > 0$, $\Phi$ is multiplicatively primitive with $p=4$ in \Cref{def:mPrim}.
    
    It remains to prove $\det v(u)=0$ for any $u\in U(2)$. By definition,
    $$v(u)=r^{-1/2}C^{1/2}(u_{11}w_1+u_{12}w_2)C^{-1/2}.$$
    So, it is equivalent to prove $\det(a w_1+bw_2)=0$ for any $a,b\in \cc$. Then,
    $$\det(a w_1+bw_2)=b\det\begin{pmatrix}a& b\\ a&b\end{pmatrix}=0$$
    yields the proposition.
\end{proof}
\begin{remark}
    Since $\det v(u)=0$ for any $u\in U(2)$, for any $p\in \nn$, $\mathcal V_1^p$ is a subset of the singular matrices. So, the set of polynomials in $(z_i)_{i=1}^{2p}$ defined by the entries of $\prod_{\ell=0}^{p-1} \left(z_{2\ell+1}v_1+z_{2\ell+2} v_2\right)$ cannot be algebraically independent. Indeed, $\det\left(\prod_{\ell=0}^{p-1} \left(z_{2\ell+1}v_1+z_{2\ell+2} v_2\right)\right)=0$. It follows that the method of proof used for \Cref{prop:xpl 1 3D} is inefficient for this example and multiplicative primitivity, as defined in \Cref{def:mPrim}, is not equivalent to $\mathcal V_1^p=M_d(\cc)$ for some $p\in\nn$.
\end{remark}

\appendix
\section{On the different notions of irreducibility}\label{app:irr}
In this section we gather some equivalent formulations of irreducibility as defined in \Cref{def:irr} and primitivity as defined in \Cref{def:prim}. These results are not new, but we did not find an appropriate published reference.
\begin{theorem}\label{thm:irr invariant subsapce}
    Let $\Phi:M_d(\cc)\to M_d(\cc)$ be completely positive with Kraus rank\footnote{We recall that the Kraus rank is the minimal $k\in \nn$ such that $\Phi$ can be written as $\Phi(X)=\sum_{i=1}^k v_i^*X v_i$ with $\{v_i\}_{i=1}^k$ a set of $d\times d$ matrices.} $r$. Let $k\in \nn$ be larger or equal to $r$.
    Let $\mu$ be a measure over $U(k)$ such that $\Phi:X\mapsto \int_{U(k)}v(u)^*Xv(u)\d\mu(u)$.
    Then, $\Phi$ is irreducible, as defined in \Cref{def:irr} or \cite[Lemma~2.1]{Evans1978}, if and only if any subspace $E$ of $\cc^d$ such that $v(u)E\subseteq E$ for $\mu$-almost every $u$ is either $\{0\}$ or $\cc^d$.
\end{theorem}
\begin{proof}
    While this result can be derived as a consequence of \cite[Lemma~2.1]{Evans1978}, we provide a direct proof.
    
    Assume that any non-zero subspace $E$ of $\cc^d$ with $v(u)E\subseteq E$ for $\mu$-almost every $u$ is such that $E=\cc^d$.  Fix $\hat x\in\pcd$ and let $E_{x}=\operatorname{linspan}\left(\{x\}\cup\{v(u_{n})\dotsb v(u_{1})x: n\in \nn, (u_{1},\dotsc,u_n)\in \supp\mu^n\}\right)$. By definition $E_x$ is nontrivial and invariant under $\mu$-almost every $v(u)$. Hence, $E_x=\cc^d$. Then, a dimensional argument leads to
    $$E_{x}=\operatorname{linspan}\left(\{x\}\cup\{v(u_{n})\dotsb v(u_{1})x: n\in \{1,\dotsc,d-1\}, (u_{1},\dotsc,u_n)\in \supp\mu^n\}\right)=\cc^d.$$
    It follows that, for any $\hat y\in \pcd$, there exist $u_1,\dotsc,u_p\in U(k)$ with $p\leq d-1$ such that either $\langle y,x\rangle\neq 0$ or $\langle y, v(u_p)\dotsb v(u_1) x\rangle \neq 0$. Then, by positivity and linearity,
    $$\langle x,(\Phi+\id_{M_d(\cc)})^{d-1}(|y\rangle \langle y|)x\rangle>0.$$
    Since $\hat x$ and $\hat y$ are arbitrary and $X\mapsto (\Phi+\id_{M_d(\cc)})^{d-1}(X)$ is linear, it follows that $\Phi$ is irreducible.
    
    For the reverse implication, assume $\Phi$ is irreducible and let $E$ be a nontrivial subspace of $\cc^d$ such that $v(u)E\subseteq E$ for $\mu$-almost every $u$, but $E\neq \cc^d$. Then, there exists $\hat y \in \pcd$ such that $y\perp E$. Since $E$ is $v(u)$-invariant, and by continuity of $u\mapsto v(u)$, for any $u_1,\dotsc, u_p\in \supp\mu$ with $p\leq d-1$, $y\perp v(u_p)\dotsb v(u_1) E=0$. It follows that $\langle x,(\Phi+\id_{M_d(\cc)})^{d-1}(|y\rangle \langle y|)x\rangle=0$ for any $x\in E$, which contradicts the irreducibility of $\Phi$.
\end{proof}

\begin{theorem}\label{thm:prim generated subspace}
    Let $\Phi$ be a quantum channel.
    Let $\mu$ be a measure over $U(k)$ such that $\Phi:X\mapsto \int_{U(k)}v(u)^*Xv(u)\d\mu(u)$.
    Then, $\Phi$ is primitive, as in \Cref{def:prim}, if and only if for any $\hat x\in \pcd$ there exists $n\in \nn$ such that
    $$\operatorname{linspan}\{v(u_{n})\dotsb v(u_{1})x: (u_{1},\dotsc,u_n)\in \supp\mu^n\}=\cc^d.$$
\end{theorem}
\begin{proof}
    For $\hat x\in\pcd$ and $n\in \nn$, let $E_{x,n}=\operatorname{linspan}\{v(u_{n})\dotsb v(u_{1})x: (u_{1},\dotsc,u_n)\in \supp\mu^n\}$.
    
    Assume $\Phi$ is primitive and let $n$ be the integer of \Cref{def:prim}. Assume $E_{x,n}\neq \cc^d$, so there exists $ \hat y\in \pcd$ such that $y\perp E_{x,n}$. It follows for any $(u_1,\dotsc,u_n)\in \supp\mu^n$, that $\langle y,v(u_n)\dotsb v(u_1)x\rangle=0$. Hence, 
    $$\langle x, \Phi^n(|y\rangle\langle y|) x\rangle=0,$$
    which contradicts the assumption that $\Phi$ is primitive. Therefore, $E_{x,n}=\cc^d$.
    
    Conversely, if $E_{x,n}=\cc^d$ then, for any $\hat y\in \pcd$, there exist $(u_1,\dotsc,u_n)\in \supp\mu^n$ such that $\langle y,v(u_n)\dotsb v(u_1)x\rangle\neq 0$. Thus,
    $$\langle x, \Phi^n(|y\rangle\langle y|) x\rangle>0.$$
    The linearity of $\Phi^n$ and the fact that $\hat x$ and $\hat y$ are arbitrary yield that $\Phi$ is primitive.
\end{proof}

\section{Space of operators induced by instruments}\label{app:instru}

\begin{lemma}\label{lem:IC instrument space}
	Let $\mathcal J$ be an informationally complete instrument, see \Cref{def:IC}, over the measurable space $(\Omega,\mathcal F)$ and assume $\Phi=\mathcal{J}(\Omega)$ is irreducible with period $m$. Then there exist $n\in \nn$ and a decomposition of $\cc^d$ into orthogonal subspaces $\{E_i\}_{i=1}^m$, such that for any $p\geq n$,
	$$\left\{\int_{\Omega^p}f(\omega_1,\dotsc,\omega_p)\mathcal J(\omega_1)\circ\dotsb\circ\mathcal J(\omega_p)(\id_{\cc^d})\d\mu^{\otimes p}(\omega_1,\dotsc,\omega_p): f \in L^{\infty}(\mu^{\otimes p})\right\}=\bigoplus_{i=1}^m\mathcal L(E_i),$$
	where we recall that $\mathcal L(E)$ is the set of endomorphisms of the finite dimensional linear space $E$.
\end{lemma}
\begin{proof}
    Let 
    $$\mathcal B_{\mathcal J}^p := \left\{\int_{\Omega^p}f(\omega_1,\dotsc,\omega_p)\mathcal J(\omega_1)\circ\dotsb\circ\mathcal J(\omega_p)(\id_{\cc^d})\d\mu^{\otimes p}(\omega_1,\dotsc,\omega_p): f \in L^{\infty}(\mu^{\otimes p})\right\}.$$
    From \cite[Theorem~4.2]{Evans1978}, there exists an orthogonal resolution of the identity $\{P_i\}_{i=1}^m$ such that $\Phi(P_i)=P_{i-1}$ with $P_{0}=P_{m}$. Since $\Phi=\mathcal J(\Omega)$ and $\mathcal J$ is $\sigma$-additive, $\mathcal J(A)(X)\leq \Phi(X)$ for any $A\in \mathcal F$ and $X\geq 0$. Then, $\mathcal J(A)(P_i)\leq P_{i-1}$ for any $i\in \{1,\dotsc,m\}$ and $A\in \mathcal F$. It follows that setting $E_i=P_i\cc^d$ for $i\in \{1,\dotsc,m\}$ implies that for any $A\in \mathcal F$ and $i\in \{1,\dotsc,m\}$, $\mathcal J(A)(\id_{\cc^d})\in \bigoplus_{\ell=1}^m \mathcal L(E_\ell)$. Hence, for any $p\in \nn$,
    $$\mathcal B_{\mathcal J}^p\subset \bigoplus_{i=1}^m \mathcal L(E_i).$$

    For the reverse inclusion, first note that $\mathcal J(\Omega)(\id_{\cc^d})=\id_{\cc^d}$ implies $\mathcal B_{\mathcal J}^{p}\subset \mathcal B_{\mathcal J}^{p+1}$. Hence, if the reverse inclusion is proved for some $p=p_0$, it is true for any $p$ larger than $p_0$. Since $\Phi$ is irreducible, so is its trace dual $\Phi^*$ -- see \cite[p.~2]{Evans1978}. Therefore, $\Phi^*$ accepts a unique fixed point $\rho>0$, $\tr(\rho)=1$ -- see \cite[Theorem~2.4]{Evans1978}. Let $(e_\ell)_{\ell=1}^d$ be the canonical basis of $\cc^d$ and $(F_{k})_{k=1}^{d^2}$ be the canonical basis of $M_d(\cc)$. Let $\psi$ be the element of $\cc^d \otimes \cc^d$ defined by $\psi=e_1\otimes e_1+\dotsb+e_d\otimes e_d$. Then $|\psi\rangle\langle \psi|=\sum_{k=1}^{d^2}F_k \otimes F_k$. It follows from \cite[Theorem~4.2]{Evans1978} that
    $$\lim_{n\to\infty}(\Phi^{mn}\otimes \id_{M_d(\cc)})(|\psi\rangle\langle \psi|)=\sum_{k=1}^{d^2}\sum_{i=1}^m \tr(F_k^* \rho_i) P_i\otimes F_{k}=\sum_{i=1}^m P_i\otimes \rho_i ,$$    
    where $\rho_1,\dotsc,\rho_m$ are positive semi-definite trace one fixed points of ${\Phi^*}^m$ such that $\supp\rho_i=E_i$. Hence, there exist $n\in \nn$ and a constant $c>0$ such that 
    $$(\Phi^{mn}\otimes \id_{M_d(\cc)})(|\psi\rangle\langle \psi|)\geq c\sum_{i=1}^mP_i\otimes P_i.$$
    That implies $$\range\big(\sum_{i=1}^m P_i \otimes P_i\big)\subseteq \range\big((\Phi^{mn}\otimes \id_{M_d(\cc)})(|\psi\rangle\langle \psi|) \big).$$
    
    We introduce Kraus operators $(v_i)_{i=1}^k$ by taking $V:\cc^d\to\cc^d\otimes\cc^k$ and $M:\Omega\to M_{k,\geq}(\cc)$ as in \Cref{def:IC}, fixing an orthonormal basis $(f_i)_{i=1}^k$ of $\cc^k$ and defining $v_i$ such that $\langle x\otimes f_i,V y\rangle=\langle x,v_iy\rangle$ for any $x,y\in\cc^d$. Then, with the index notation
    $$\underline{i} := (i_1, \ldots, i_{mn}) \in \{1,\ldots,k\}^{mn} , \qquad v_{\underline{i}} := v_{i_1} \ldots v_{i_{mn}} ,$$
    we have
    \begin{align*}
    \bigoplus_{i=1}^mE_i\otimes E_i&=\range(\sum_{i=1}^m P_i\otimes P_i)\\
    &\subseteq\range\big((\Phi^{mn}\otimes \id_{M_d(\cc)})(|\psi\rangle\langle \psi|) \big) \\
    &\quad= \range \Bigg( \sum_{\underline{i}\in\{1,\dotsc,k\}^{mn}} (v_{\underline{i}} \otimes \id_{\cc^d})^* |\psi\rangle \langle\psi| (v_{\underline{i}} \otimes \id_{\cc^d}) \Bigg)\\
    &\qquad\subseteq \operatorname{linspan}\{ (v_{\underline{i}} \otimes \id_{\cc^d})^* \psi : \underline{i} \in \{1,\dotsc,k\}^{mn} \}.
    \end{align*}
    Now, $\langle \phi_1\otimes\overline{\phi_2}, (a\otimes \id_{\cc^d})\psi\rangle=\langle \phi_1,a\phi_2\rangle$ for any $a\in M_d(\cc)$, $\phi_1,\phi_2\in \cc^d$ where $\overline{\phi_2}$ is the complex conjugate of $\phi_2$ in the canonical basis. Thus,
    $$\bigoplus_{i=1}^m \mathcal L(E_i)\subseteq\operatorname{linspan}\{ v_{\underline{i}}^* : \underline{i} \in \{1,\dotsc,k\}^{mn} \}.$$
    Finally, we take the step from Kraus operators $v_{\underline{i}}^*$ to operators in $\mathcal B_{\mathcal J}^{mn} $: If for any $X \in \bigoplus_{i=1}^m \mathcal L(E_i)$, we can show that there exists some $g_X \in L^{\infty}(\mu^{\otimes mn})$ such that
    \begin{equation} \label{eq:X_rewrting_equation}
    X = \int_{\Omega^{mn}}g_X(\omega_1,\dotsc,\omega_{mn})\mathcal J(\omega_1)\circ\dotsb\circ\mathcal J(\omega_{mn})(\id_{\cc^d})\d\mu^{\otimes mn}(\omega_1,\dotsc,\omega_{mn}),
    \end{equation}
    then $\bigoplus_{i=1}^m \mathcal L(E_i) \subset \mathcal B_{\mathcal J}^{mn}$, and the proof is done.
    
    To construct $g_X$, we write $ X = \id_{\cc^d}^* X $ and define the coefficient vector $\phi_X\in (\cc^k)^{\otimes mn}$ (so $\phi_X(\underline{i}) \in \cc$ for $\underline{i} \in \{1,\ldots,k\}^{mn}$) such that $X=\sum_{\underline{i}}\overline{\phi_X(\underline{i})} v_{\underline{i}}$ and $\phi_{\id}\in (\cc^k)^{\otimes mn}$ such that $\id_{\cc^d}=\sum_{\underline{i}} \overline{\phi_\id(\underline{i})} v_{\underline{i}}$. 
    Now, note that the POVM $M$ being informationally complete (\emph{i.e.} satisfying \Cref{eq:IC POVM}) implies that for any $\ell \in \nn$,
    $$\left\{\int_{\Omega^\ell}g(\omega_1,\dotsc,\omega_\ell)M(\omega_1)\otimes\dotsb\otimes M(\omega_\ell)\d\mu^{\otimes \ell}(\omega_1,\dotsc,\omega_\ell): g\in L^\infty(\mu^{\otimes \ell})\right\}=M_k(\cc)^{\otimes \ell}.$$
    We can therefore find a function $g_X\in L^{\infty}(\mu^{\otimes mn})$ such that
    $$|\phi_\id\rangle\langle\phi_X|=\int_{\Omega^{mn}}g_X(\omega_1,\dotsc,\omega_{mn})M(\omega_1)\otimes\dotsb\otimes M(\omega_{mn})\d\mu^{\otimes mn}(\omega_1,\dotsc,\omega_{mn}).$$
    Then, using that for any $Y\in M_d(\cc)$,
    $$\mathcal J(\omega_1)\circ\dotsb\circ\mathcal J(\omega_{mn})(Y)=V_{mn}^*\left(Y\otimes M(\omega_1)\otimes\dotsb\otimes M(\omega_{mn})\right)V_{mn} ,$$
    with $V_{mn}:\cc^d\to\cc^d\otimes(\cc^k)^{\otimes mn}$ defined by
    $$V_{mn}x=\sum_{\underline{i}} (v_{\underline{i}} x)\otimes f_{\underline{i}} ,$$
    where $f_{\underline{i}} := f_{i_1}\otimes\dotsb\otimes f_{i_{mn}}$, the right-hand side of \Cref{eq:X_rewrting_equation} amounts to
    \begin{align*}
    &\int_{\Omega^{mn}}g_X(\omega_1,\dotsc,\omega_{mn})\mathcal J(\omega_1)\circ\dotsb\circ\mathcal J(\omega_{mn})(\id_{\cc^d})\d\mu^{\otimes mn}(\omega_1,\dotsc,\omega_{mn}) \\
    &= V_{mn}^* (\id_{\cc^d} \otimes |\phi_\id\rangle\langle\phi_X|) V_{mn}
    = \sum_{\underline{i}, \underline{j}} \phi_{\id}(\underline{j}) \overline{\phi_X(\underline{i})} v_{\underline{j}}^* v_{\underline{i}}
    = \id_{\cc^d}^* X
    = X ,
    \end{align*}
    as claimed. So indeed, $\bigoplus_{i=1}^m \mathcal L(E_i) \subset \mathcal B_{\mathcal J}^{p_0}$ with $p_0 = mn$.
\end{proof}

\section{GAP Measures}\label{app:GAP}

In this appendix we briefly discuss \textit{Gaussian adjusted projected (GAP) measures}, an important class of probability distributions on the sphere of a separable Hilbert space $\mathcal{H}$. Roughly speaking, for any density matrix $\rho$ on $\mathcal{H}$, in terms of information gained by a measurement, $\GAP_\rho$ is the most spread out distribution over the sphere $S(\mathcal{H})$ that has density matrix $\rho$. 

We motivate the use of GAP measures, give an explicit construction and alternative definitions (in finite dimensions) and state some crucial properties. For further details we refer the reader to \cite{GLTZ06,GLMTZ16} and \cite{PhDthesisVogel}.

\subsection{Motivation} GAP measures were first introduced by Jozsa, Robb, and Wootters~\cite{JRW94} in an information theoretic context. They considered the ``accessible information'' of an ensemble, a quantifier for the maximal amount of classical information that can be extracted from it, and showed that under the constraint that its density matrix is given by $\rho$, this information is minimized by $\GAP_\rho$. 
More precisely, let $\nu$ be a probability measure on $\pcd$ and $M$ a POVM on the space $\Omega$. The joint law of the system state and measurement outcome is $\d\mathfrak{p}_{\nu,M}(\hat x,\omega)=\langle x,M(\d \omega)x\rangle\d\nu(\hat x)$. Its marginal on $\pcd$ is $\nu$ and, on $\Omega$, it is $\mathbb Q_{\rho_\nu}(\d\omega)=\tr(M(\d\omega)\rho_\nu)$ with $\rho_\nu=\ee_\nu(|x\rangle\langle x|)$. Then the mutual information between $(\hat x,\omega)\mapsto \hat x$ and $(\hat x,\omega)\mapsto \omega$ with respect to $\mathfrak{p}_{\nu,M}$ is the Kullback-Leibler divergence (or relative entropy)
$$I(\mathfrak{p}_{\nu,M})=D_{KL}(\mathfrak{p}_{\nu,M}|\nu\otimes\mathbb Q_{\rho_\nu}).$$
In \cite{JRW94}, the authors show that $\GAP_\rho$ is the measure $\nu$ minimizing
$$\sup_{M}I(\mathfrak{p}_{\nu,M})$$
over the set of probability measures $\nu$ such that $\rho_\nu=\rho$. Namely,
$$\text{GAP}(\rho)=\operatornamewithlimits{argmin}_{\nu : \rho_\nu=\rho}\sup_{M}I(\mathfrak{p}_{\nu,M}).$$
They named the measure \textit{Scrooge measure}, referring to Ebenezer Scrooge, the very stingy protagonist of Charles Dicken's novella \textit{A Christmas Carol} (1843), as $\GAP_\rho$ is ``particularly stingy with its information''. 

A couple of years later, Goldstein, Lebowitz, Mastrodonato, Tumulka and Zanghì~\cite{GLTZ06,GLMTZ16} realized that GAP measures also naturally occur in quantum statistical mechanics: They showed that if $\rho$ is a canonical density matrix, \emph{i.e.}, if it is of the form
\begin{align*}
	\rho_{\mathrm{can}} = \frac{1}{Z} e^{-\beta H},
\end{align*}
where $H$ is the Hamiltonian, $\beta$ the inverse temperature and $Z$ a normalization constant, then $\GAP_\rho$ describes the thermal equilibrium distribution of the wave function and can be seen as a quantum analogue of the canonical ensemble of classical statistical mechanics. More precisely, they show that for bipartite systems, typically the wave function of the subsystem (in the sense of the \textit{conditional wave function} \cite{DGZ92, GN99, GLTZ06}) is $\GAP_{\rho_{\mathrm{can}}}$-distributed where $\rho_{\mathrm{can}}$ is a canonical density matrix for the subsystem. We remark that similar considerations can also be made for grand-canonical density matrices~\cite{ITV26}. 

\subsection{Construction} We now give the construction of GAP measures to which its acronym refers and for simplicity we restrict ourselves to finite-dimensional Hilbert spaces $\mathcal H$. Note that a similar construction is also possible for separable Hilbert spaces, see~\cite{Tum20}.

The starting point for the construction is a Gaussian measure, which we then adjust and finally project onto the sphere. 

Let $\rho$ be a density matrix on $\mathcal{H}$, and $\Psi^G$ a Gaussian complex random vector of zero mean and covariance $\rho$. Let $\mathrm{G}_\rho$ denote the probability measure on $\mathcal{H}$ describing the distribution of $\Psi^{\mathrm{G}}$.
As we require that $\GAP_\rho$ has expected density matrix $\rho$, \emph{i.e.}, that
\begin{align*}
	\rho_{\mathrm{GAP}_{\rho}} := \int_{S(\mathcal{H})} |\psi\rangle\langle\psi| \d\mathrm{GAP}_{\rho}(\psi) = \rho,
\end{align*}
we cannot simply project $G_\rho$ onto $S(\mathcal{H})$; this would result in the wrong density matrix. Thus, some adjustment is necessary, and we define the \textit{Gaussian adjusted measure} $\text{GA}_\rho$ by
\begin{align*}
	\d\mathrm{GA}_\rho(\psi) = \|\psi\|^2 \d G_\rho(\psi),
\end{align*}
i.e., we adjust the density with the factor $\|\psi\|^2$. Note that $\mathbb{E}\|\Psi^G\|^2 = \tr(\rho)=1$ ensures that $\text{GA}_\rho$ is also a probability measure on $\mathcal{H}$.

In the last step, we project the measure $\text{GA}_\rho$ to the sphere $S(\mathcal{H})$. Let $\Psi^{\mathrm{GA}}$ be a $\text{GA}_\rho$-distributed vector. We then define $\GAP_\rho$ as the distribution of the random vector
\begin{align*}
	\Psi^{\mathrm{GAP}} := \frac{\Psi^{\mathrm{GA}}}{\|\Psi^{\mathrm{GA}}\|}.
\end{align*}
We remark that as $\Psi^{\mathrm{GA}}$ is continuously distributed, $\|\Psi^{\mathrm{GA}}\|\neq 0$ with probability 1 and thus $\Psi^{\mathrm{GAP}}$ is well-defined. Moreover, by definition $\GAP_\rho$ satisfies $\rho_{\mathrm{GAP}_\rho}=\rho$.

\subsection{Alternative definitions} In finite dimension, there are different ways to define GAP measures. We discuss two that are relevant to us. The first one is defining $\GAP_{\rho}$ by its density with respect to the uniform measure. More precisely, suppose that $\rho$ has finite rank $r=\dim\mathrm{supp}(\rho)$. Then the density of $\GAP_{\rho}$ relative to the uniform probability measure $u$ on $S(\mathrm{supp}(\rho))$ is given by~\cite[Equation (18)]{GLTZ06}:
\begin{align}\label{eq:density of GAP}
\frac{\d \GAP_{\rho}}{\d u}(\psi) = \frac{r}{\det\rho_+} \langle\psi|\rho_+^{-1}|\psi\rangle^{-r-1},
\end{align}
where $\rho_+$ denotes the restriction of $\rho$ to its support $\mathrm{supp}(\rho)$.

For the second alternative definition let $\mathcal{K}$ be a copy of $\mathcal H$ and let $\Phi\in S(\mathcal{H}\otimes\mathcal K)$ with $\tr_{\mathcal K} |\Phi\rangle\langle\Phi|=\rho$, where $\tr_{\mathcal K}$ denotes the partial trace over $\mathcal K$. Moreover, let $u_{\mathcal K}$ be the uniform probability measure on $S(\mathcal K)$ and let $\Psi_{\mathcal K} \sim \mu_{\mathcal K}$, where $\mu_{\mathcal K}$ is the measure with density
\begin{align*}
    \frac{\d\mu_{\mathcal K}}{\d u_{\mathcal K}}(\psi_{\mathcal K}) = \dim(\mathcal K)\|(\id_{\mathcal H}\otimes \langle\psi_{\mathcal K}|)|\Phi\rangle\|^2.
\end{align*}
Then, $\GAP_{\rho}$ is the distribution of the random vector
\begin{align*}
\Psi := \frac{(\id_{\mathcal H}\otimes \langle\Psi_{\mathcal K}|)|\Phi\rangle}{\|(\id_{\mathcal H}\otimes \langle\Psi_{\mathcal K}|)|\Phi\rangle\|},
\end{align*}
see again \cite[Section~1.4]{GLMTZ16}. As the authors there point out, the measure $\mu_{\mathcal K}$ can be thought of as the distribution of a quantum measurement of $\id_{\mathcal H} \otimes M_{\mathcal{K}}$ on a system in the pure state $\Phi$, where $M_{\mathcal K}$ is the POVM defined by $\d M_{\mathcal{K}}(\psi) = \dim(\mathcal K)|\psi\rangle\langle\psi| \d u_{\mathcal K}(\psi)$.

\subsection{Properties} We collect a couple of useful properties of GAP measures in the following proposition. For the proofs we refer to \cite{GLTZ06,GLMTZ16}; see also Section~1.5.3 in \cite{PhDthesisVogel} for a longer collection of properties of GAP measures.
\begin{proposition}\label{prop:GAP}
Let $\rho$ be a density matrix on a separable Hilbert space $\mathcal{H}$. Then,
\begin{enumerate}[label=(\alph*)]
    \item $\GAP_{\rho}$ has density matrix $\rho$, \emph{i.e.}, $\rho_{\GAP_{\rho}}=\rho$.
    \item The map $\rho\mapsto\GAP_{\rho}$ is covariant: For any unitary operator $U$ on $\mathcal{H}$ and any measurable subset $A\subset S(\mathcal H)$,
\begin{align*}
    \GAP_{\rho}(\{U^*|\psi\rangle : |\psi\rangle\in A\})=\GAP_{U\rho U^*}(A).
\end{align*}
    \item $\GAP_{\rho}$ is invariant under global phase changes, \emph{i.e.}, for any $\theta\in\mathbb{R}$ and any measurable subset $A\subset S(\mathcal{H})$,
\begin{align*}
\GAP_{\rho}(e^{i\theta} A) = \GAP_{\rho}(A). 
\end{align*}
\end{enumerate}
\end{proposition}

\section{SageMath 10 code listing}\label{app:code}
Listing of the SageMath 10 code used to check for algebraic independence in the examples of \Cref{sec:xpl 3D}. It is also available in the Git repository \href{https://plmlab.math.cnrs.fr/tbenoist/multiplicative-primitivity}{plmlab.math.cnrs.fr/tbenoist/multiplicative-primitivity}.
\begin{lstlisting}[language=Python]
#!/usr/bin/env sage
##############################
# MultPrim.sage: Checking for multiplicative primitivity
# Exact computations made on the algebraic field (QQbar) by default.
# Language: SageMath 10
# Date: 2026-02-07
# Authors: Tristan Benoist, Sascha Lill, Cornelia Vogel
# Version: 0.2
##############################

# Computation of Bp
def compute_Bp(Kraus,p,base_ring=QQbar):
    """
    Returns the matrix of polynomials B_p=V_1^p given an integer p 
    and a list of Kraus square matrices.
    """
    dim_Kraus=Kraus[0].ncols() # Dimension of the Kraus matrices
    nKraus=len(Kraus) # Number of Kraus operators

    R=PolynomialRing(base_ring,'z',nKraus*p)

    Bp=matrix(R,dim_Kraus,dim_Kraus, sparse=False)
    Bp=identity_matrix(R,dim_Kraus)
    for k in range(p):
        V=zero_matrix(R,dim_Kraus)
        for i in range(nKraus):
            V=V+R.gen(nKraus*k+i)*Kraus[i]
        Bp=Bp*V
    return Bp

# Are the entries of Bp algebraically independent?
def is_algebraically_indep(list_poly,coords=None):
    """
    If True is returned, then the polynomials in the list list_poly 
    are algebraically independent.
    Returns also, as a list, the coordinates where the Jacobian 
    has been evaluated.
    """
    polynomial_ring=list_poly[0].parent()
    variables=polynomial_ring.gens()
    ring=polynomial_ring.base_ring()

    if len(coords)!=polynomial_ring.ngens():
        raise ValueError("Mismatch number of coordinates for evaluation")

    if coords!=None:
        matrix_jacobi=jacobian(list_poly,variables)(coords)
    else: # If no coordinates are given, generate random coordinates.
        coords=[ring(choice(range(1,11))) for i in range(len(variables))]
        matrix_jacobi=jacobian(list_poly,variables)(coords)

    if matrix_jacobi.rank()<len(list_poly):
        return False, coords
    else:
        return True, coords

# Is multiplicative primitivity true?
def is_mPrim_rank_jacobi_Bp(Kraus,p,coords=None,base_ring=QQbar):
    """
    Returns the coordinates where the Jacobian has been evaluated as a list 
    and True if Bp=V_1^p is the full matrix space.
    p is an integer and Kraus a list of square matrices.
    """
    
    Bp=compute_Bp(Kraus,p,base_ring)
    return is_algebraically_indep(Bp.list(),coords)

def roots_of_jacobian_det_Bpx(Kraus,p,i_minor=0,base_ring=QQbar):
    """
    Returns the zeros (as a list) of the Jacobian determinant of Bp*x 
    with the entries of the x as the variables.
    i_minor is the first index of the minor of the Jacobian matrix 
    whose determinant is computed.
    """

    Bp=compute_Bp(Kraus,p,base_ring)
    nKraus=len(Kraus)
    dim=Kraus[0].ncols()

    list_var=var(['x'+str(i) for i in range(dim)])

    R_Bx=PolynomialRing(base_ring,'z',nKraus*p+dim)
    Bx=Bp*vector([R_Bx.gen(nKraus*p+i) for i in range(dim)])

    list_poly=Bx.list()
    variables=list_poly[0].parent().gens()[:-dim]

    matrix_jacobi=jacobian(list_poly,variables)

    ddet=det(matrix_jacobi[:,i_minor:i_minor+dim])

    return solve([ddet(list(variables)+list(list_var))==0], list_var)


#############################################
# If the script is not imported but directly called
if __name__=='__main__':
    # List of Kraus operators (list of square matrices of the same size)
    Kraus_xpl1=[matrix([[0,1,0],[1,0,1],[0,0,0]]),
                        matrix([[0,0,0],[0,-1,0],[1,0,-1]])]
    Kraus_xpl2=[matrix([[1,1,0],[-1,1,0],[0,0,0]]),
                                matrix([[0,0,1],[0,0,1],[1,0,0]])]

    # Example 1
    p=8
    mPrim, coords=is_mPrim_rank_jacobi_Bp(Kraus_xpl1,p,
                        [i%
    if mPrim:
        print("In Example 1, multiplicative primitivity is verified.")
        print("Jacobian evaluated in: {}\n".format(vector(coords)))
    else:
        print("Example 1 possibly not multiplicatively primitive.\n")

    # Example 2
    p=4
    roots=roots_of_jacobian_det_Bpx(Kraus_xpl2,p)
    print("In Example 2, the roots of the Jacobian determinant of a 3x3 minor 
        of Bp*x are:")
    show([vector([x.rhs() for x in root]) for root in roots])
    
#EOF
\end{lstlisting}
The code output is:
{\footnotesize
\begin{verbatim}
In Example 1, multiplicative primitivity is verified.
Jacobian evaluated in: (1, 2, 3, 1, 2, 3, 1, 2, 3, 1, 2, 3, 1, 2, 3, 1)

In Example 2, the roots of the Jacobian determinant of a 3x3 minor of Bp*x are:
[(r1, r2, -(r1*z1*z5*z7 - ((r1 - r2)*z1*z4 - (r1 + r2)*z0*z5)*z6)/((z1*z4 + z0*z5)*z7)),
 (r3, r4, -(2*r3*z2*z5*z7 - (2*(r3 - r4)*z2*z4 - (r3 + r4)*z3*z5)*z6)/((2*z2*z4 + z3*z5)*z7)),
 (r5, r6, -1/2*(2*r6*z4*z6 + r5*z5*z7)/(z4*z7))]
\end{verbatim} }

\bigskip
\textbf{Acknowledgements.} The research of TB was partly funded by ANR project DYNACQUS,
grant number ANR-24-CE40-5714.
SL acknowledges financial support by the European Union (ERC \textsc{FermiMath} nr.~101040991 and ERC \textsc{MathBEC} nr.~101095820). CV~was supported by the Deutsche Forschungsgemeinschaft (DFG, German Research Foundation) -- TRR 352 -- Project-ID 470903074. Moreover, C.V.~acknowledges financial support by the ERC Starting Grant ``FermiMath" No.~101040991 and the ERC Consolidator Grant ``RAMBAS'' No. 10104424, funded by the European Union. 

Views and opinions expressed are those of the authors and do not necessarily reflect those of the European Union or the European Research Council Executive Agency. Neither the European Union nor the granting authority can be held responsible for them.

\textbf{Conflicts of interest.} The authors declare no conflict of interest with respect to the present
article.

    \printbibliography
     
\end{document}